\documentclass[11pt,letterpaper]{article} 
\usepackage[margin=1in]{geometry}
\usepackage{setspace}
\usepackage[numbers]{natbib}
\usepackage[hidelinks]{hyperref}
\usepackage{bm}
\usepackage[utf8]{inputenc} 
\usepackage[T1]{fontenc}    
\usepackage{lmodern}
\usepackage{fix-cm}
\rmfamily
\DeclareFontShape{T1}{lmr}{b}{sc}{<->ssub*cmr/bx/sc}{}
\DeclareFontShape{T1}{lmr}{bx}{sc}{<->ssub*cmr/bx/sc}{}
\usepackage{booktabs}       
\usepackage{amsfonts}       
\usepackage{nicefrac}       
\usepackage{microtype}      
\usepackage{booktabs}
\usepackage{xcolor}
\usepackage{xfrac}
\usepackage{color}
\usepackage{amsmath,amsthm,amsfonts,amssymb,bbm}
\usepackage{thm-restate}
\usepackage{enumitem}
\usepackage{cleveref}
\usepackage{xspace}
\usepackage{tikz}
\usepackage{subcaption}

\hypersetup{
	colorlinks,
	allcolors=magenta
}

\usepackage[ruled]{algorithm2e}
\usepackage{comment} 

\SetAlFnt{\small}
\SetAlCapFnt{\small}
\SetAlCapNameFnt{\small}
\SetAlCapHSkip{0pt}
\SetKwInput{KwEl}{Elicitation Rule}
\SetKwInput{KwAgg}{Aggregation Rule}
\SetKwInput{KwComm}{Communication Complexity}
\SetKwInput{KwDist}{Distortion}
\SetKwInput{KwVoting}{Voting Rule}
\SetKwInput{KwSampling}{Steps}
\SetKwInput{KwCommunicate}{Communicate}
\SetKwInput{KwSample}{Sample}
\IncMargin{-\parindent}

\newtheorem{theorem}{Theorem}

\newtheorem{lemma}{Lemma}
\newtheorem{proposition}{Proposition}
\newtheorem{corollary}{Corollary}
\theoremstyle{definition}
\newtheorem{definition}{Definition}
\newtheorem{example}{Example}

\newcount\Comments  
\newcommand{\kibitz}[2]{\ifnum\Comments=1{\color{#1}{#2}}\fi}

\newcommand{\ceil}[1]{\lceil #1 \rceil}
\newcommand{\floor}[1]{\lfloor #1 \rfloor}

\renewcommand{\hat}{\widehat}
\newcommand{\E}{\mathop{\mathbb{E}}}

\newcommand{\set}[1]{\{#1\}}

\newcommand{\bbR}{\mathbb{R}}
\newcommand{\bbN}{\mathbb{N}}
\newcommand{\calL}{\mathcal{L}}
\newcommand{\calM}{\mathcal{M}}

\renewcommand{\bar}{\overline}
\renewcommand{\succeq}{\succcurlyeq}

\DeclareMathOperator*{\argmin}{arg\,min}
\DeclareMathOperator*{\argmax}{arg\,max}

\DeclareMathOperator{\dist}{distortion}
\DeclareMathOperator{\fair}{fairness}
\DeclareMathOperator{\SC}{SC}
\let\top\relax
\DeclareMathOperator{\top}{top}
\DeclareMathOperator{\plu}{plu}
\DeclareMathOperator{\uni}{uni}
\DeclareMathOperator{\veto}{veto}

\newcommand{\vsigma}{\sigma}
\newcommand{\elec}{\mathcal{E}}
\newcommand{\defvotese}[3]{D_{#1}^{#3}(#2)}
\newcommand{\restr}[2]{\ensuremath{\left.#1\right|_{#2}}}

\newcommand{\plumatching}{\textsc{PluralityMatching}\xspace}
\newcommand{\unimatching}{\textsc{UniformMatching}\xspace}
\newcommand{\randdict}{\textsc{RandomDictatorship}\xspace}
\newcommand{\smartdict}{\textsc{SmartDictatorship}\xspace}

\title{Resolving the Optimal Metric Distortion Conjecture}
\author{
	Vasilis Gkatzelis\\Drexel University\\\texttt{gkatz@drexel.edu}
	\and
	Daniel Halpern\\University of Toronto\\\texttt{daniel.halpern@mail.utoronto.ca}
	\and
	Nisarg Shah\\University of Toronto\\\texttt{nisarg@cs.toronto.edu}
}
\date{}

\begin{document}\sloppy\allowdisplaybreaks
\pagenumbering{Alph}

\maketitle

\begin{abstract}
We study the following metric distortion problem: there are two finite sets of points, $V$ and $C$, that lie in the same metric space, and our goal is to choose a point in $C$ whose total distance from the points in $V$ is as small as possible. However, rather than having access to the underlying distance metric, we only know, for each point in $V$, a ranking of its distances to the points in $C$. We propose algorithms that choose a point in $C$ using only these rankings as input and we provide bounds on their \emph{distortion} (worst-case approximation ratio). A prominent motivation for this problem comes from voting theory, where $V$ represents a set of voters, $C$ represents a set of candidates, and the rankings correspond to ordinal preferences of the voters.

A major conjecture in this framework is that the optimal deterministic algorithm has distortion $3$. We resolve this conjecture by providing a polynomial-time algorithm that achieves distortion $3$, matching a known lower bound. We do so by proving a novel lemma about matching voters to candidates, which we refer to as the \emph{ranking-matching lemma}. This lemma induces a family of novel algorithms, which may be of independent interest, and we show that a special algorithm in this family achieves distortion $3$. We also provide more refined, parameterized, bounds using the notion of $\alpha$-decisiveness, which quantifies the extent to which a voter may prefer her top choice relative to all others. Finally, we introduce a new randomized algorithm with improved distortion compared to known results, and also provide improved lower bounds on the distortion of all deterministic and randomized algorithms.
\end{abstract}

\thispagestyle{empty}
\clearpage

\pagenumbering{arabic}
\setcounter{page}{1}
\section{Introduction}\label{sec:intro}

In the metric distortion problem, we are given two finite (potentially non-disjoint) sets $V$ and $C$ that lie in the same metric space. If we had access to the distance metric $d$ over $V \cup C$, we could compute a point $c^* \in C$ that minimizes its total distance to the points in $V$, i.e., $c^* \in \argmin_{c \in C} \sum_{i \in V} d(i,c)$. However, rather than having access to the exact distances, we are given, for each point $i \in V$, a ranking $\sigma_i$ over the points in $C$ by their distance to $i$ (with ties broken arbitrarily). Given only this collection of rankings $\vsigma = (\sigma_i)_{i \in V}$, we need to design an algorithm (a social choice rule) that picks a point $\hat{c} \in C$ (possibly in a randomized fashion), and we are interested in its worst-case approximation ratio $\E[\sum_{i \in V} d(i,\hat{c})] / \sum_{i \in V} d(i,c^*)$, where the expectation is over randomization of the algorithm. This worst-case approximation ratio is known as the \emph{distortion} of the algorithm in this framework. Note that the need for approximation stems not from computational constraints, but from the lack of full information. Our main question is: \emph{To what extent can we minimize the total distance using only ordinal information regarding pairwise distances?}

An interesting motivating example for this question stems from voting theory, where $V$ is a set of voters, $C$ is a set of candidates, and the distance between voters and candidates corresponds to the distance of their ideologies with respect to different issues. This model, inspired by spatial models of voting from the political science literature, such as the Downsian proximity model~\cite{enelow1984spatial}, has recently become popular in the computer science literature~\cite{anshelevich2018approximating,skowron2017social,munagala2019improved}. An alternative, possibly more concrete, interpretation of these distances can also be derived through the well-studied problem of facility location~\cite{charikar1999guha,JMMSV03,cohen2019polynomial}. In this case $V$ corresponds to a set of customers and $C$ is the set of feasible locations for a new facility to be opened; the distance measure between customers and locations could correspond to the actual physical distance between them, or the time or cost required for a customer to reach each location. In both cases, the most common way to elicit the preferences of the voters, or customers, is in the form of a ranking of the available options (e.g., see~\cite[Part I]{BCEL+16} or~\cite{marianov2008facility}), and the goal is to pick an outcome that minimizes the total distance.

Centuries of research on voting theory has produced a plethora of voting rules that choose an outcome using the ranked preferences of voters as input. In recent work, starting with \citet{anshelevich2015approximating}, researchers have analyzed the performance of these voting rules using the metric distortion framework described above, which enables a quantitative comparison among them.
\citet{anshelevich2015approximating} showed that no deterministic algorithm
can achieve a distortion better than $3$, while a popular voting rule known as Copeland's rule comes close, by achieving distortion $5$. Informally, Copeland's rule measures the strength of each candidate $c$ as the number of candidates $c'$ defeated by $c$ in a pairwise election (i.e., the candidates $c'$ that are ranked below $c$ by the majority of voters), and picks a candidate with the greatest strength. They also analyzed the distortion of other popular voting rules such as plurality, Borda count, $k$-approval, and veto, and showed that their distortion increases linearly with the number of voters or candidates. 
Subsequently, \citet{skowron2017social} analyzed another popular voting rule, known as single transferable vote (STV), and showed that its distortion grows logarithmically in the number of candidates. 

The first deterministic algorithm to break the bound of $5$ was very recently proposed by \citet{munagala2019improved}, who designed a novel deterministic algorithm that achieves a distortion of $2+\sqrt{5} \approx 4.236$. Their algorithm does not even require access to voters' full ordinal preferences. It takes as input only the weighted tournament graph, which contains, for every pair of candidates, the number of voters that prefer one over another. They proved that $2+\sqrt{5}$ is the best distortion of any deterministic algorithm with this limited input, and they conjectured that with full ordinal preferences, the best possible distortion should be $3$, outlining a sufficient condition for achieving distortion $3$. However, they were unable to prove that there always exists a candidate satisfying their sufficient condition. 
Independently, \citet{kempe2020duality} derived the same bound of $2+\sqrt{5}$ using a linear programming duality framework, and provided several alternative formulations of the sufficient condition outlined by \citet{munagala2019improved}. He used the linear programming duality framework to show that the distortion of two social choice rules, the ranked pairs and Schulze's rule, grows as the square root of the number of candidates, thus disproving a conjecture by \citet{anshelevich2015approximating} that ranked pairs might have distortion $3$ (this conjecture was also earlier disproved by \citet{goel2017metric}). 
The main result of our paper is a culmination of this line of work in the form of a novel deterministic algorithm that guarantees the optimal distortion of 3. 

\subsection{Our Results \& Techniques}\label{sec:contrib}
The technical core of our paper is in Section~\ref{sec:rm-lemma}, where we prove the \emph{Ranking-Matching Lemma}, an elementary lemma regarding the existence of fractional perfect matchings within a family of vertex-weighted bipartite graphs induced by an instance of our problem. Specifically, given a problem instance and some candidate $a$,
we get a bipartite graph whose vertices correspond to the set of voters on one side, and the set of candidates on the other. 
The vertex of each voter $i$ has some non-negative weight $p(i)$ and the vertex of each candidate $c$ has some non-negative weight $q(c)$ such that the total weight of all voters $\sum_{i \in V} p(i)$ and the total weight of all candidates $\sum_{c \in C} q(c)$ are both normalized to $1$. 
An edge $(i,c)$ between the vertices of some voter $i$ and some candidate $c$ (potentially equal to $a$) exists if and only if candidate $a$ is (weakly) closer voter $i$ than candidate $c$, indicating that voter $i$ would (weakly) prefer candidate $a$ to $c$. Intuitively, the more edges this graph has, the more appealing candidate $a$ appears, since more voters would prefer $a$ over other candidates. For example, if the graph is complete, then $a$ must be the top choice of all voters. Following this intuition, we propose a notion of fractional perfect matching in this graph so that, if such a matching exists for the graph of candidate $a$, then $a$ is an ``appealing'' candidate. Specifically, this matching suggests that for every other candidate $c$, there are sufficiently many \emph{distinct} voters who prefer $a$ to $c$. 
The Ranking-Matching Lemma shows that for \emph{any} problem instance and \emph{any} weight vectors $p$ and $q$, there always exists some candidate $a$ whose graph admits a fractional perfect matching, i.e., there always exists at least one appealing candidate. 

Armed with this non-trivial result, we then leverage it in order to achieve a good distortion bound. A natural social choice rule would be to select one of the candidates that satisfy the aforementioned matching criterion. Note that the Ranking-Matching Lemma leaves open the choice of weights of voters and candidates (given by $p$ and $q$). These weights can be used to quantify the relative significance of voters and candidates, and each choice yields a different algorithm. As we discuss in \Cref{sec:disc}, some of these algorithms could be of independent interest, but in this paper we focus on the case where the weights of all voters are the same, and those of the candidates are proportional to their plurality score (i.e., the number of voters that rank them first). We observe that for this set of weights, the Ranking-Matching Lemma also implies the existence of an \emph{integral} perfect matching in a different bipartite graph whose vertices on both sides correspond to voters. For a given candidate $a$, we refer to this graph as the integral domination graph of $a$ and denote it by $G(a)$, and an edge $(i,j)$ between the vertices of voters $i$ and $j$ exists if voter $i$ weakly prefers $a$ to the top choice candidate of voter $j$. 

We define a novel voting rule, \plumatching, which returns a candidate $a$ whose integral domination graph $G(a)$ has a perfect matching, and the Ranking-Matching Lemma implies that such a candidate always exists. We then prove our main result in Section~\ref{sec:det-dist}, which is that \plumatching achieves the optimal distortion bound of $3$. 

Our approach is inspired by the recent work of \citet{munagala2019improved}, who used integral perfect matchings in bipartite graphs to define the \emph{matching uncovered set}. Specifically, given a \emph{pair} of candidates $a$ and $b$, they define a bipartite graph $G(a,b)$, in which both sides of vertices correspond to the voters, and edge $(i,j)$ between voters $i$ and $j$ exists if there is some candidate $c$ for which voter $i$ weakly prefers $a$ to $c$ and voter $j$ weakly prefers $c$ to $b$. The matching uncovered set is the set of candidates $a$ such that for every candidate $b$ the $G(a,b)$ graph has a perfect matching. They proved that every candidate in this set has distortion at most $3$, but were unable to answer whether this set is always non-empty. Our integral domination graph $G(a)$ of candidate $a$ is a subgraph of $G(a,b)$ for every candidate $b$. Thus, a candidate $a$ for which $G(a)$ has a perfect matching is also in the matching uncovered set, and the existence of such a candidate is guaranteed by our Ranking-Matching Lemma, implying that the matching uncovered set is always non-empty. 

While the non-emptiness of the matching uncovered set, combined with the argument of \citet{munagala2019improved}, implies a distortion bound of $3$, we provide a direct proof of this bound that is much shorter and simpler because it uses the existence of a perfect matching in the much sparser
integral domination graph (thus using the full strength of the Ranking-Matching Lemma). More significantly, as we show in Section~\ref{subsec:connection-to-the-mus} this proof also allows us to achieve improved distortion bounds for interesting special cases which the matching uncovered set does not meet.

Having settled the question of distortion for the class of deterministic algorithms, an exciting open problem is to also achieve optimal distortion bounds for \emph{randomized} algorithms. Prior to this work, the best known distortion bound achieved by a randomized algorithm was $3-2/m$~\cite{kempe2020communication}, which is better than what any deterministic algorithm can achieve. In Section~\ref{sec:rand-dist}, we take a step toward a better understanding of the power and limitations of randomization in this context. Specifically, we propose a new randomized algorithm, which we term \smartdict, that only depends on the plurality scores of the candidates. We show that it matches the distortion bound of $3-2/m$ for general metric spaces, but has improved distortion in special cases. Further, we show that among randomized algorithms that depend only on the plurality scores, \smartdict has optimal distortion, even in the aforementioned special cases. We also provide an improved lower bound that applies to randomized algorithms that use full preference rankings.

As the number of candidates grows, no known distortion bound is better than 3, and we know that no randomized algorithm can achieve better than 2; the problem of identifying the optimal bound in $[2,3]$ for $m \to \infty$, and the optimal bound for a fixed value of $m$ in general, remains open.

Throughout the paper, our distortion bounds are presented as a function of the number of candidates in the instance, but we also parameterize our bounds using the notion of $\alpha$-decisiveness due to \citet{anshelevich2017randomized}. Given a parameter value $\alpha\in [0,1]$, an instance is $\alpha$-decisive if every voter's distance from her top candidate is at most $\alpha$ times her distance from any other candidate. For a motivating example, consider the peer selection problem, where the voters are choosing a candidate among themselves as a representative; in this case, every voter's distance from herself is zero, so this is a special case of $\alpha=0$. Among other results, we are the first to establish non-trivial (i.e.\ greater than $1$) lower bounds on distortion for this special case for deterministic as well as randomized algorithms.

The paper is organized as follows. \Cref{sec:model} presents the model, \Cref{sec:rm-lemma} presents the Ranking-Matching Lemma, and \Cref{sec:det-dist,sec:rand-dist} study the distortion of deterministic and randomized algorithms, respectively. \Cref{sec:related} provides an overview of related work not covered above, \Cref{sec:disc} contains a discussion of our results and open directions, and the appendix contains deferred proofs and additional observations. 
\section{Model}\label{sec:model}
Let $V = \set{1,\ldots,n}$ be a set of $n$ \emph{voters}, and $C$ be a set of $m$ \emph{candidates}. Throughout the paper, we use letters $i,j,k$ to index voters, and letters $a,b,c$ to index candidates. Also, we use $\Delta(V)$ and $\Delta(C)$ to denote the set of probability distributions over $V$ and $C$, respectively; for simplicity we also use this notation for vectors of nonnegative weights that add up to 1, even when these weights do not correspond to probabilities. 

The voters and candidates are located in a common metric space $(\calM,d)$, where $d : \calM \times \calM \to \bbR_{\ge 0}$ is the metric (distance function) satisfying the following conditions: (1) Positive definiteness: $\forall x,y \in \calM$, $d(x,y) \ge 0$, and $d(x,y) = 0$ if and only if $x=y$; (2) Symmetry: $\forall x,y \in \calM$, $d(x,y) = d(y,x)$; and (3) Triangle inequality: $\forall x,y,z \in \calM$, $d(x,y) + d(y,z) \ge d(x,z)$

For simplicity of notation, we extend $d$ to operate directly on the voters and the candidates instead of the points of the metric space where they are located. Thus, we will write $d(x,y)$, for $x,y \in V \cup C$, to mean the distance between the points where voters/candidates $x$ and $y$ are located. Given a metric, $d$, the \emph{social cost} of a candidate $c$ is defined as $\SC(c,d) = \sum_{i \in V} d(i,c)$; we simply write $\SC(c)$ when the metric $d$ is clear from the context.

\paragraph{Preference Profiles.}
Given sets $V, C$ and a metric $d$, for each voter $i\in V$ we get a linear order, or \emph{preference ranking}, $\sigma_i$ over the candidates in $C$, where $i$ ranks candidates in non-decreasing order of their distance from her. We say that candidate $c$ \emph{defeats} candidate $c'$ in vote $i$, denoted $c \succ_i c'$, if voter $i$ ranks $c$ above $c'$; we say that $c$ \emph{weakly defeats} $c'$, denoted $c \succeq_i c'$, if $c = c'$ or $c \succ_i c'$. We say that $d$ is consistent with the preference ranking $\sigma_i$ of voter $i$, denoted $d \triangleright \sigma_i$, if $d(i,c) \le d(i,c')$ for all $c,c'$ such that $c \succ_i c'$. Note that candidates that are equidistant to a voter can be ranked arbitrarily by the voter. We refer to a profile of preference rankings, $\vsigma = (\sigma_i)_{i \in V}$, as the \emph{preference profile}, and we say that $d$ is consistent with the preference profile $\vsigma$, denoted $d \triangleright \vsigma$, if $d \triangleright \sigma_i$ for each $i \in V$. 

For a given $\vsigma$, let $\top(i) \in C$ denote the candidate ranked highest in $\sigma_i$, and let the plurality score of candidate $a$, i.e., the number of voters who rank $a$ as their top choice, be $\plu(a) = |\set{i \in V : \top(i) = a}|$. With slight abuse of notation, we also write $\plu(D) = \sum_{a \in D} \plu(a)$ for all $D \subseteq C$. 

\paragraph{Social Choice Rules and Distortion.}
Given an instance $(V,C,d)$, our main goal in this paper is to output a candidate $c\in C$ that minimizes the social cost $\SC(c,d)$. However, instead of having full access to the metric $d$, we have access to a tuple $\elec = (V,C,\vsigma)$, which we refer to as an \emph{election}, and all we know is that $d$ is consistent with $\vsigma$.

Let $\calL(C)$ denote the set of linear orders (rankings) over $C$.
A \emph{social choice rule} (or simply a rule) is a function $f: \cup_{n \ge 1} \calL(C)^n \to \Delta(C)$. We say that the rule is \emph{deterministic} if it always returns a distribution with singleton support. Note that we seek rules that operate on preference profiles consisting of any number of votes, but over a fixed set of candidates. 
The \emph{distortion} of a social choice rule $f$, denoted $\dist(f)$, is the worst-case approximation ratio it provides to social cost:
\[
\dist(f) = \sup_{\vsigma}\ \sup_{d\: :\: d\: \triangleright\: \vsigma}\ \frac{\E[\SC(f(\vsigma),d)]}{\min_{c \in C} \SC(c,d)},
\]
where the expectation is over possible randomization in the social choice rule. Note that there is no restriction on the number of voters in the preference profile, so we are interested in the distortion as a function of  the number of candidates. This is because in real-world applications of voting, it is reasonable to expect the number of candidates, $m$, to be small --- finer distortion bounds as a function of $m$ are therefore important --- but the number of votes may be large.

\paragraph{Decisive Voters.}
We further refine the distortion bounds using the $\alpha$-decisiveness framework of \citet{anshelevich2017randomized}. For $\alpha \in [0,1]$, voter $i$ is called $\alpha$-decisive if $d(i,\top(i)) \le \alpha \cdot d(i,c)$, $\forall c \neq \top(i)$. In words, the distance of voter $i$ from her top choice is at most $\alpha$ times her distance from any other candidate. We say that the metric space is $\alpha$-decisive if every voter is $\alpha$-decisive. For a given value of $\alpha$, when we analyze the distortion of social choice rules for $\alpha$-decisive metric spaces, we take the worst case only over metric spaces that are $\alpha$-decisive and for which $d \triangleright \vsigma$.  

 Note that for $\alpha = 0$ the distance of every voter from her top-choice candidate is zero. This captures settings where each voter is a candidate herself, and ranks herself at the top. For instance, the \emph{peer selection} setting, in which $n$ agents submit ranked votes over themselves and each agent ranks herself first, is a special case of $\alpha=0$; thus, the upper bounds (and some lower bounds) presented in our paper for $\alpha = 0$ apply to this interesting special case as well. When $\alpha=1$, we retrieve the general, unrestricted, metric space.

\section{Ranking-Matching Lemma}\label{sec:rm-lemma}

We begin by proving our key contribution, which is an elementary lemma about matching preference rankings, i.e., voters, to candidates. We then use this lemma to define a novel deterministic social choice rule, and in the next section show that this rule has distortion at most $3$, thus resolving the open conjecture. First, we need to introduce some terminology.

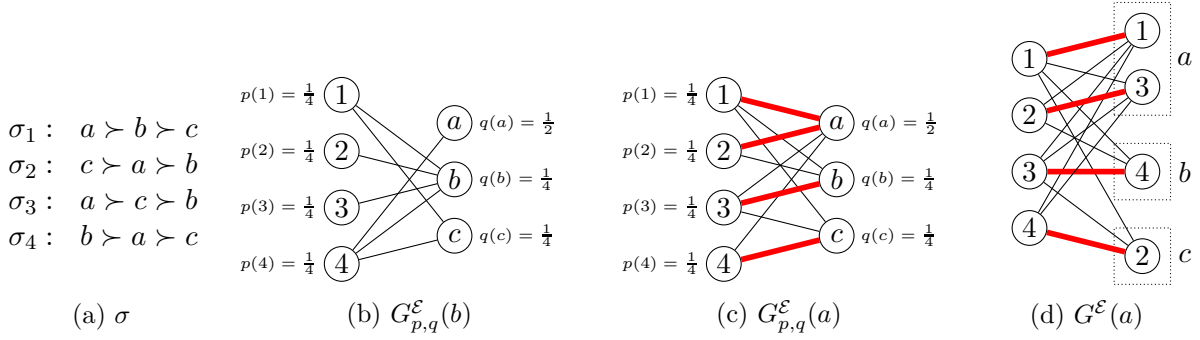
\begin{figure}
	\begin{subfigure}[b]{0.19\textwidth}
		\centering
		\[\begin{array}{ll}
			  \sigma_1 :& a \succ b \succ c\\
			  \sigma_2 :& c \succ a \succ b\\
			  \sigma_3 :& a \succ c \succ b\\
			  \sigma_4 :& b \succ a \succ c
		\end{array}\]
		\caption{$\vsigma$}
		\label{fig:rankings}
	\end{subfigure}
	\begin{subfigure}[b]{0.29\textwidth}
		\begin{tikzpicture}[scale=.75]
			\centering
			\tikzstyle{main_node} = [circle,fill=white,draw,minimum size=1.2em,inner sep=0pt]

			\node[main_node] (1) at (0, 0) {$1$};
			\node[main_node] (2) at (0, -1) {$2$};
			\node[main_node] (3) at (0, -2) {$3$};
			\node[main_node] (4) at (0, -3) {$4$};
			\node[main_node] (a) at (2, -.5) {$a$};
			\node[main_node] (b) at (2, -1.5) {$b$};
			\node[main_node] (c) at (2, -2.5) {$c$};

			\draw (1) node [left=.5em] {\tiny $p(1) = \frac{1}{4}$};
			\draw (2) node [left=.5em] {\tiny $p(2) = \frac{1}{4}$};
			\draw (3) node [left=.5em] {\tiny $p(3) = \frac{1}{4}$};
			\draw (4) node [left=.5em] {\tiny $p(4) = \frac{1}{4}$};
			
			\draw (a) node [right=.5em] {\tiny $q(a) = \frac{1}{2}$};
			\draw (b) node [right=.5em] {\tiny $q(b) = \frac{1}{4}$};
			\draw (c) node [right=.5em] {\tiny $q(c) = \frac{1}{4}$};

			\draw (1)--(b);
			\draw (1)--(c);
			\draw (2)--(b);
			\draw (3)--(b);
			\draw (4)--(a);
			\draw (4)--(b);
			\draw (4)--(c);
		\end{tikzpicture}
		\caption{$G^\elec_{p, q}(b)$}
		\label{fig:b-domination}
	\end{subfigure}
	\begin{subfigure}[b]{0.29\textwidth}
		\centering
		\begin{tikzpicture}[scale=.75]
			\tikzstyle{main_node} = [circle,fill=white,draw,minimum size=1.2em,inner sep=0pt]

			\node[main_node] (1) at (0, 0) {$1$};
			\node[main_node] (2) at (0, -1) {$2$};
			\node[main_node] (3) at (0, -2) {$3$};
			\node[main_node] (4) at (0, -3) {$4$};
			\node[main_node] (a) at (2, -.5) {$a$};
			\node[main_node] (b) at (2, -1.5) {$b$};
			\node[main_node] (c) at (2, -2.5) {$c$};

			\draw (1) node [left=.5em] {\tiny $p(1) = \frac{1}{4}$};
			\draw (2) node [left=.5em] {\tiny $p(2) = \frac{1}{4}$};
			\draw (3) node [left=.5em] {\tiny $p(3) = \frac{1}{4}$};
			\draw (4) node [left=.5em] {\tiny $p(4) = \frac{1}{4}$};

			\draw (a) node [right=.5em] {\tiny $q(a) = \frac{1}{2}$};
			\draw (b) node [right=.5em] {\tiny $q(b) = \frac{1}{4}$};
			\draw (c) node [right=.5em] {\tiny $q(c) = \frac{1}{4}$};

			\draw (1)--(b);
			\draw (1)--(c);

			\draw (2)--(b);
			\draw (3)--(a);

			\draw (4)--(a);
			\draw (3)--(c);
			\draw (1)--(a) [draw=red, line width=.75mm];
			\draw (2)--(a) [draw=red, line width=.75mm];
			\draw (3)--(b) [draw=red, line width=.75mm];
			\draw (4)--(c) [draw=red, line width=.75mm];
		\end{tikzpicture}
		\caption{$G^\elec_{p, q}(a)$}
		\label{fig:a-domination}
	\end{subfigure}
	\begin{subfigure}[b]{0.19\textwidth}
		\centering
		\begin{tikzpicture}[scale=.75]
			\tikzstyle{main_node} = [circle,fill=white,draw,minimum size=1.2em,inner sep=0pt]

			\node[] () at (-.8, 0) {};

			\node[main_node] (1l) at (0, 0) {$1$};
			\node[main_node] (2l) at (0, -1) {$2$};
			\node[main_node] (3l) at (0, -2) {$3$};
			\node[main_node] (4l) at (0, -3) {$4$};

			\node[main_node] (1r) at (2, .5) {$1$};
			\node[main_node] (3r) at (2, -.5) {$3$};

			\node[main_node] (4r) at (2, -2) {$4$};

			\node[main_node] (2r) at (2, -3.5) {$2$};

			\draw[densely dotted] (1.5, -1)--(2.5,-1)--(2.5,1)--(1.5,1)--cycle;
			\node at (2.75, 0) {$a$};

			\draw[densely dotted] (1.5, -2.5)--(2.5,-2.5)--(2.5,-1.5)--(1.5,-1.5)--cycle;
			\node at (2.75, -2) {$b$};

			\draw[densely dotted] (1.5, -4)--(2.5,-4)--(2.5,-3)--(1.5,-3)--cycle;
			\node at (2.75, -3.5) {$c$};

			\draw (1l)--(2r);
			\draw (1l)--(3r);
			\draw (1l)--(4r);
			\draw (2l)--(1r);

			\draw (2l)--(4r);
			\draw (3l)--(1r);
			\draw (3l)--(2r);
			\draw (3l)--(3r);

			\draw (4l)--(1r);
			\draw (1l)--(1r) [draw=red, line width=.75mm];
			\draw (2l)--(3r) [draw=red, line width=.75mm];
			\draw (3l)--(4r) [draw=red, line width=.75mm];
			\draw (4l)--(2r) [draw=red, line width=.75mm];
			\draw (4l)--(3r);
		\end{tikzpicture}
		\caption{$G^\elec(a)$}
		\label{fig:a-integral}
	\end{subfigure}
	\caption{Domination graph examples}
	\label{fig:matching-example}
\end{figure}

\begin{definition}
	Given an election $\elec = (V,C,\vsigma)$, and (normalized) weight vectors $p \in \Delta(V)$ and $q \in \Delta(C)$, the $(p,q)$-\emph{domination graph of candidate $a$} is the vertex-weighted bipartite graph $G^{\elec}_{p,q}(a) = (V,C,E_a,p,q)$, where $(i,c) \in E_a$ if and only if $a \succeq_i c$. Vertex $i \in V$ has weight $p(i)$ and vertex $c \in C$ has weight $q(c)$. We will drop $\elec$ from notation when it is clear from the context.
\end{definition}

For example, \Cref{fig:rankings} shows the preference rankings of four voters ($1, 2, 3, 4$) over three candidates ($a, b, c$). \Cref{fig:b-domination} and \Cref{fig:a-domination} show the $(p,q)$-domination graphs of candidates $b$ and $a$, respectively, with $p(1)=p(2)=p(3)=p(4)=\frac{1}{4}$ and $q(a)=\frac{1}{2}$, $q(b)=q(c)=\frac{1}{4}$.

\begin{definition}
We say that the $(p,q)$-domination graph of candidate $a$, i.e., $G^{\elec}_{p,q}(a)$, \emph{admits a fractional perfect matching} if there exists a weight function $w : E_a \to \bbR_{\ge 0}$ such that the total weight of edges incident on each vertex equals the weight of the vertex, i.e.,
for each $i \in V$ we have $\sum_{c\in C: (i,c) \in E_a} w((i,c)) = p(i)$, and for each $c \in C$ we have $\sum_{i\in V: (i,c) \in E_a} w((i,c)) = q(c)$.
\end{definition}

With slight abuse of notation, let us denote $p(S) = \sum_{i \in S} p(i)$ for all $S \subseteq V$ and $q(D) = \sum_{a \in D} q(a)$ for all $D \subseteq C$. Given an election $\elec = (V,C,\vsigma)$, a candidate $a \in C$, and a subset of voters $S \subseteq V$, define $\defvotese{a}{S}{\elec} = \set{c \in C : a \succeq_i c, \exists i \in S}$ to be the set of candidates that $a$ weakly defeats in at least one vote in $S$. Note that this is precisely the set of neighbors of $S$ in $G^{\elec}_{p,q}(a)$ (for any $p$ and $q$). A simple generalization of Hall's condition shows that the existence of a fractional perfect matching in $G^{\elec}_{p,q}(a)$ is equivalent to the set of neighbors of $S$ having at least as much weight as $S$ itself, for all $S \subseteq V$. For the sake of completeness, we provide a short proof in Appendix~\ref{sec:proof-halls}. 

\begin{lemma}\label{lem:general-halls-condition}
	For any election $\elec = (V,C,\vsigma)$, weight vectors $p \in \Delta(V)$ and $q \in \Delta(C)$, and candidate $a \in C$, $G^{\elec}_{p,q}(a)$ admits a fractional perfect matching if and only if $q(\defvotese{a}{S}{\elec}) \ge p(S)$ for all $S \subseteq V$, and whether $G^{\elec}_{p,q}(a)$ admits a fractional perfect matching can be checked in strongly polynomial time. 
\end{lemma}

Using this lemma for the instance in Figure~\ref{fig:matching-example}, we can verify that $G^{\elec}_{p,q}(b)$ does not admit a fractional perfect matching by letting $S=\{2,3\}$, for which $p(S)=\frac{1}{2}$, and observing that $\defvotese{b}{S}{\elec}=\{b\}$ with $q(b)=\frac{1}{4}$. On the other hand, $G^{\elec}_{p,q}(a)$ does admit a fractional perfect matching, indicated using thick red edges in \Cref{fig:a-domination}.

Note that if, in a given election, there exists some candidate $a$ such that $G^{\elec}_{p,q}(a)$ admits a fractional perfect matching, then by \Cref{lem:general-halls-condition}, in \emph{any} subset of votes $S$, $a$ weakly defeats a set of candidates with total weight at least that of $S$. This intuitively makes candidate $a$ a compelling choice. In fact, in \Cref{sec:det-dist} we show that for appropriate choices of $p$ and $q$, this also implies low distortion. However, it is far from obvious that such a candidate $a$ is guaranteed to exist. Our key lemma, presented next, shows that there will always exist such a candidate in \emph{any} election $\elec$ and for \emph{any} weight vectors $p$ and $q$. 

While we only need to prove the lemma for the choices of $p$ and $q$ that we use in \Cref{sec:det-dist} to settle the conjecture, we prove it for all possible choices of $p$ and $q$ for two reasons. First, our proof is inductive, and the proof for given choices of $p$ and $q$ in a given election requires the lemma to hold for \emph{all} choices of $p$ and $q$ in smaller elections. Second, as we discuss in \Cref{sec:integral-matching,sec:disc}, other choices of $p$ and $q$ also lead to interesting social choice rules, which may be of independent interest.

\begin{lemma}[\textsc{Ranking-Matching Lemma}]\label{lem:rm-lemma}
	For any election $\elec$, and weight vectors $p \in \Delta(V)$ and $q \in \Delta(C)$, there exists a candidate $a$ whose $(p,q)$-domination graph $G^{\elec}_{p,q}(a)$ admits a fractional perfect matching, and such a candidate can be computed in strongly polynomial time.
\end{lemma}
\begin{proof}
	Suppose that the lemma is not true. Let $\elec = (V,C,\vsigma)$ be an election with the smallest number of voters for which the lemma is violated for any weight vectors $p \in \Delta(V)$ and $q \in \Delta(C)$. Fix such $p$ and $q$. By \Cref{lem:general-halls-condition}, we have that for each candidate $a \in C$, there exists $S \subseteq V$ for which $q(\defvotese{a}{S}{\elec}) < p(S)$. We refer to such a set $S$ as a \emph{counterexample for $a$}. We say that a counterexample for $a$ is \emph{minimal} if no strict subset of it is a counterexample for $a$. For each $a \in C$, fix an arbitrary minimal counterexample $X_a \subseteq V$. 
	
	Consider the set of candidates $C^* = \argmax_{c \in C} p(X_c)$ whose minimal counterexamples have the largest weight under $p$. Define a partial order $R$ over $C^*$, whereby, for $b,c \in C^*$, $b R c$ if and only if $X_b \supseteq X_c$ and $b \succ_i c$ for all $i \in X_c$. Let $a \in C^*$ be an arbitrary maximal element of this partial order. Note that the counterexample for $a$ has the highest weight, and no other candidate $b$ with this property ``dominates'' $a$ in the sense of being better than $a$ under relation $R$.

	Fix this candidate $a$. For brevity of notation, write $X = X_a$ and $D = \defvotese{a}{X_a}{\elec}$. We use the standard notation $\bar{X} = V \setminus X$ and $\bar{D} = C \setminus D$. The next lemma shows that for any candidate $b$ that $a$ does not weakly defeat in its minimal counterexample, the minimal counterexamples of $a$ and $b$ are incomparable in the sense that neither is a subset of the other.
	
	\begin{lemma}\label{lem:non-inclusion-prob}
		For each $b \in \bar{D}$, we have $X \setminus X_b \neq \emptyset$ and $X_b \setminus X \neq \emptyset$.
	\end{lemma}
	\begin{proof}
		Fix $b \in \bar{D}$. By the definition of $D$, this implies $b \succ_i a$ for all $i \in X$. 
		
		First, we show that $X \setminus X_b \neq \emptyset$. Suppose this is false, so $X \subseteq X_b$. Then, we have $p(X_b) \ge p(X)$. However, since $a \in C^* = \argmax_{c \in C} p(X_c)$, we also have $b \in C^*$. Further, since $X \subseteq X_b$ and $b \succ_i a$ for all $i \in X$, we have $b R a$. This contradicts the fact that $a$ is a maximal element of $C^*$ under partial order $R$.
		
		Next, we show that $X_b \setminus X \neq \emptyset$. Suppose this is not true, so $X_b \subseteq X$. Since we have already established that $X \setminus X_b \neq \emptyset$, we know $X_b \neq X$. Hence, $X_b \subsetneq X$. Since $X$ is a minimal counterexample for $a$, $X_b$ is not a counterexample for $a$. Thus, $\defvotese{a}{X_b}{\elec} \ge p(X_b)$. However, since $b \succ_i a$ for all $i \in X$, we have that $q(\defvotese{b}{X_b}{\elec}) \ge q(\defvotese{a}{X_b}{\elec}) \ge p(X_b)$, which contradicts the fact that $X_b$ is a counterexample for $b$. 
	\end{proof}
	
We now use \Cref{lem:non-inclusion-prob} to prove \Cref{lem:rm-lemma}. We now make four claims: (1) $D \neq \emptyset$, (2) $q(\bar{D}) > 0$  (which implies $\bar{D} \neq \emptyset$), (3) $p(X) > 0$ (which implies $X \neq \emptyset$), and (4) $\bar{X} \neq \emptyset$. The first claim is true because $a \in D$. For the next two claims, recall that $X$ is a counterexample for $a$. Hence, $0 \le q(D) < p(X) \leq 1$, which implies both $q(\bar{D}) > 0$ and $p(X) > 0$. Finally, for the last claim, suppose for contradiction that $\bar{X} = \emptyset$, i.e., $X = V$. Then, $X_b \setminus X = \emptyset$ for all $b$. Thus, by \Cref{lem:non-inclusion-prob}, $\bar{D} = \emptyset$, which contradicts claim (2). 

Next, we show that every candidate $c \in \bar{D}$ weakly defeats candidates of sufficient total weight in its minimal counterexample $X_c$, i.e., that $q(\defvotese{c}{X_c}{\elec}$ is sufficiently large. Later, we will use this to identify a specific candidate $b \in \bar{D}$ and show that $q(\defvotese{b}{X_b}{\elec} \ge p(X_b)$, contradicting the fact that $X_b$ is a counterexample for $b$. For all $c \in \bar{D}$, we have that
	\begin{align}
	q(\defvotese{c}{X_c}{\elec}) &= q(\defvotese{c}{X_c}{\elec} \cap D) + q(\defvotese{c}{X_c}{\elec} \cap \bar{D}) \nonumber\\
	&\ge q(\defvotese{c}{X_c \cap X}{\elec} \cap D) + q(\defvotese{c}{X_c \cap \bar{X}}{\elec} \cap \bar{D}) \nonumber\\
	&\ge q(\defvotese{a}{X_c \cap X}{\elec}) + q(\defvotese{c}{X_c \cap \bar{X}}{\elec} \cap \bar{D}) \nonumber\\
	&\ge p(X_c \cap X) + q(\defvotese{c}{X_c \cap \bar{X}}{\elec} \cap \bar{D}). \label{eqn:q-Xc}
	\end{align}
	Here, the first transition holds because for any set $T$, sets $T \cap D$ and $T \cap \bar{D}$ form a partition of $T$. The second transition holds because $\defvotese{c}{X_c \cap T}{\elec} \subseteq \defvotese{c}{X_c}{\elec}$ for any set $T$.

	The third transition holds because $\defvotese{a}{X_c \cap X}{\elec} \subseteq \defvotese{c}{X_c \cap X}{\elec} \cap D$: every candidate in $\defvotese{a}{X_c \cap X}{\elec}$ (i.e. weakly defeated by $a$ in $X_c \cap X$) is both in $D$ (because it is also weakly defeated by $a$ in $X$) and in  $\defvotese{c}{X_c \cap X}{\elec}$ (because it is also weakly defeated by $c$ in $X_c \cap X$ as $c \succ_i a$ for all $i \in X$).

	Finally, for the last transition, note that $X_c \cap X \subsetneq X$ due to \Cref{lem:non-inclusion-prob}. Since $X$ is a minimal counterexample for $a$, $X_c \cap X$ is not a counterexample for $a$. Thus, $q(\defvotese{a}{X_c \cap X}{\elec}) \ge p(X_c \cap X)$.

	Our goal now is to find some candidate $b \in \bar{D}$ and use \Cref{eqn:q-Xc} to derive a contradiction to the fact that $X_b$ is a counterexample for $b$ (that is, show that $q(\defvotese{b}{X_b}{\elec}) \ge p(X_b)$). We consider two cases, $p(\bar{X}) = 0$ and $p(\bar{X}) > 0$, and show that for each case we can identify such a candidate $b \in \bar{D}$.
	
First, suppose $p(\bar{X}) = 0$. Fix an arbitrary $b \in \bar{D}$; such a candidate exists because we know $\bar{D} \neq \emptyset$. Since $p(X) = 1$, we have that $p(X_b \cap X) = p(X_b)$. Substituting this in \Cref{eqn:q-Xc} already yields $q(\defvotese{b}{X_b}{\elec}) \ge p(X_b)$, which contradicts the fact that $X_b$ is a counterexample for $b$.
	
The more interesting case is when $p(\bar{X}) > 0$, which implies $\bar{X} \neq \emptyset$. In this case, consider the restricted election $\bar{\elec} = (\bar{X}, \bar{D}, \restr{\vsigma}{\bar{X},\bar{D}})$, where $\restr{\vsigma}{\bar{X},\bar{D}}$ denotes the preference profile $\vsigma$ restricted to the preferences of the voters in $\bar{X}$ over the candidates in $\bar{D}$. Note that $\bar{\elec}$ is a valid election because $\bar{X} \neq \emptyset$ and $\bar{D} \neq \emptyset$, as shown above. Further, $\bar{\elec}$ has fewer voters than $\elec$ because $X \neq \emptyset$. Therefore, $\bar{\elec}$ must satisfy \Cref{lem:rm-lemma} for any weight vectors, since $\elec$ was chosen as an election with the smallest number of voters for which \Cref{lem:rm-lemma} fails to hold. In particular, choose the re-normalized weight vectors $p' \in \Delta(\bar{X})$ and $q' \in \Delta(\bar{D})$ as follows:
	\[
	p'(i) = p(i)/p(\bar{X}), \forall i \in \bar{X}, ~~\text{ and }~~ q'(c) = q(c)/q(\bar{D}), \forall c \in \bar{D}.
	\]
	Note that these are well-defined because we have already shown that $q(\bar{D}) > 0$, and we are in the case $p(\bar{X}) > 0$. Since \Cref{lem:rm-lemma} holds for $\bar{\elec}$ with weight vectors $p'$ and $q'$, there exists some candidate $b \in \bar{D}$ such that
	$p'(S) \leq q'(\defvotese{b}{S}{\bar{\elec}})$ for all $S \subseteq \bar{X}$. In particular, for $S = X_b \cap \bar{X}$, we have,
	\[
	p'(X_b \cap \bar{X}) \le q'(\defvotese{b}{X_b \cap \bar{X}}{\bar{\elec}}) = q'(\defvotese{b}{X_b \cap \bar{X}}{\elec} \cap \bar{D}),
	\]
	where the last transition holds because $\bar{\elec}$ is a restriction of $\elec$ to the voters in $\bar{X}$ and the candidates in $\bar{D}$, which implies $\defvotese{b}{X_b \cap \bar{X}}{\bar{\elec}}  = \defvotese{b}{X_b \cap \bar{X}}{\elec} \cap \bar{D}$. Expanding the definitions of $q'$ and $p'$, we get
	\[
	\frac{q(\defvotese{b}{X_b \cap \bar{X}}{\elec} \cap \bar{D})}{q(\bar{D})} \geq \frac{p(X_b \cap \bar{X})}{p(\bar{X})}
	\ \Rightarrow\  q(\defvotese{b}{X_b \cap \bar{X}}{\elec} \cap \bar{D}) \ge p(X_b \cap \bar{X}) \cdot \frac{q(\bar{D})}{p(\bar{X})} \ge p(X_b \cap \bar{X}),
	\]
	where the last step follows because $X$ is a counterexample for $a$, which implies $q(D) < p(X)$, i.e., $q(\bar{D}) > p(\bar{X})$. Substituting this bound in \Cref{eqn:q-Xc}, we get
	\[
	q(\defvotese{b}{X_b}{\elec}) \ge p(X_b \cap X) + p(X_b \cap \bar{X}) = p(X_b),
	\]
	which is the desired contradiction to the fact that $X_b$ is a counterexample for $b$. 
	
	Finally, to compute a desired candidate $a$, we can simply iterate over all $m$ candidates, and check whether the $(p,q)$-domination graph of each admits a fractional perfect matching, which can be done in strongly polynomial time due to \Cref{lem:general-halls-condition}.
\end{proof}

\subsection{Weight Vectors and Integral Matching}\label{sec:integral-matching}
Note that \Cref{lem:rm-lemma} leaves open the choice of weight vectors $p$ and $q$, and thus induces a novel family of (deterministic) social choice rules: for any given $p$ and $q$, the corresponding social choice rule returns, on a given election $\elec$, an arbitrary candidate $a$ whose $(p,q)$-domination graph $G^{\elec}_{p,q}(a)$ admits a fractional perfect matching. We now discuss two interesting choices of $p$ and $q$. Perhaps the most straightforward choice is that of uniform weight vectors.

\begin{definition}[Rule \unimatching]
	Given an election $\elec = (V,C,\vsigma)$, \unimatching returns a candidate $a$ (ties broken arbitrarily) whose $(p^{\uni},q^{\uni})$-domination graph
	admits a fractional perfect matching, where $p^{\uni}(i) = 1/n$ for each $i \in V$ and $q^{\uni}(c) = 1/m$ for each $c \in C$.
\end{definition}

For any candidate $a$ returned by this rule, \Cref{lem:general-halls-condition} implies that $\nicefrac{|\defvotese{a}{S}{\elec}|}{|C|} \ge \nicefrac{|S|}{|V|}$ for all $S \subseteq V$. Succinctly, this can be summarized as follows: \emph{In any $x\%$ of the votes, for any $x$, candidate $a$ weakly defeats at least $x\%$ of the candidates.} To the best of our knowledge, ours is the first result establishing that such a candidate always exists. We believe this could be of independent interest; for a discussion on how this may be connected to another major open conjecture in voting, see \Cref{sec:disc}. The rule we are interested in, however, is induced by a different choice of weight vectors. 

\begin{definition}[Rule \plumatching]
	Given an election $\elec = (V,C,\vsigma)$, \plumatching returns a candidate $a$ (ties broken arbitrarily) whose $(p^{\uni},q^{\plu})$-domination graph
	admits a fractional perfect matching, where $p^{\uni}(i) = 1/n$ for each $i \in V$ and $q^{\plu}(c) = \plu(c)/n$ for each $c \in C$. 
\end{definition}

For any candidate $a$ returned by this rule, \Cref{lem:general-halls-condition} implies that $\plu(\defvotese{a}{S}{\elec}) \ge |S|$ for all $S \subseteq V$, i.e., in any subset of votes $S$, candidate $a$ weakly defeats a set of candidates with total plurality score at least $|S|$. Note that, in the domination graphs in Figure~\ref{fig:matching-example}, the weights used are, in fact, $p^{\uni}$ and $q^{\plu}$, which suggests that \plumatching could return candidate $a$ in that case.

With these weight vectors, we can equivalently view the domination graph a bit differently. Instead of each voter $i$ on the left and each candidate $c$ on the right having weights $1/n$ and $\plu(c)/n$, respectively, we can let them have integral weights $1$ and $\plu(c)$, respectively. Then, we can replace each node $c$ on the right with weight $\plu(c)$ by $\plu(c)$ many nodes, each representing a unique voter whose top choice is $c$, that have weight $1$ each and are connected to the same set of nodes on the left as $c$ was. This leads to a bipartite graph whose vertices on both sides correspond to the voters.
\begin{definition}
	Given an election $\elec = (V,C,\vsigma)$ and a candidate $a \in C$, define, with slight abuse of notation, the \emph{integral domination graph} of candidate $a$ as the bipartite graph $G^{\elec}(a) = (V,V,E_a)$, where $(i,j) \in E_a$ if and only if $a \succeq_i \top(j)$. 
\end{definition}

Returning to the example in \Cref{fig:matching-example}, the graph $G^{\elec}(a)$ appears in \Cref{fig:a-integral}. The only edges missing are $(2, 2)$ and $(4, 4)$, since $\top(2)=c$ and $c \succ_2 a$, and $\top(4)=b$ and $b \succ_4 a$. Comparing the graph $G^{\elec}(a)$ in \Cref{fig:a-integral} to the graph $G^\elec_{p^{\uni},q^{\plu}}(a)$ in \Cref{fig:a-domination}, note that candidate $a$ was replaced by voters $1$ and $3$ (who rank $a$ at the top), candidate $b$ was replaced by voter $4$, and candidate $c$ was replaced by voter $2$ (also indicated in the figure, using dotted boxes).

It is easy to check that Hall's condition for the existence of an \emph{(integral) perfect matching} in the integral domination graph is identical to the condition outlined above for the candidate returned by \plumatching.
Therefore, a direct consequence of \Cref{lem:rm-lemma} is that there always exists some candidate $a$ such that $G^{\elec}(a)$ admits a perfect matching, and \plumatching returns one such candidate. A perfect matching in \Cref{fig:a-integral} is denoted with thick red lines. Comparing this to \Cref{fig:a-domination}, it is easy to see how this integral matching is actually derived directly from the fractional matching in $G^\elec_{p^{\uni},q^{\plu}}(a)$. Note that we could have also chosen edges $(1, 3)$ and $(2, 1)$ instead of $(1, 1)$ and $(2, 3)$ in \Cref{fig:a-integral}.

\begin{corollary}\label{cor:integral-domination}
	Let $\elec = (V,C,\vsigma)$ be any election. For each candidate $a \in C$, $G^{\elec}(a)$ admits a perfect matching if and only if $G^{\elec}_{p^{\uni},q^{\plu}}(a)$ admits a fractional perfect matching. Consequently, there exists a candidate $a \in C$ whose integral domination graph $G^{\elec}(a)$ admits a perfect matching, i.e., for which there exists a bijection $M : V \to V$ satisfying $a \succeq_i \top(M(i))$ for each $i \in V$. 
\end{corollary}

In the next section, we use this guarantee to achieve an optimal distortion bound.
\section{The Distortion of Deterministic Social Choice Rules}\label{sec:det-dist}
\usetikzlibrary{shapes.multipart}

We begin by proving that \plumatching has distortion at most $2+\alpha$ for $\alpha$-decisive metric spaces. Since all metric spaces are $1$-decisive, this settles a major conjecture that the optimal deterministic rule has distortion $3$~\cite{anshelevich2015approximating,munagala2019improved}, matching a lower bound of $3$ established by \citet{anshelevich2015approximating}. The proof crucially leverages the fact that \plumatching returns a candidate whose integral domination graph has a perfect matching (\Cref{cor:integral-domination}). The proof is rather short, as the heavy lifting is already done in proving the Ranking-Matching Lemma (\Cref{lem:rm-lemma}). 

\begin{theorem}\label{thm:det-upper-bound}
	For every $m \ge 3$ and $\alpha \in [0,1]$, \plumatching has distortion $2 + \alpha$ for $\alpha$-decisive metric spaces.
\end{theorem}
\begin{proof}
	We begin with the upper bound. Let $(\calM,d)$ be an $\alpha$-decisive metric space, and let $\vsigma$ be the induced preference profile. Let $a$ be the candidate returned by \plumatching, and let $b$ be any other candidate. We want to show that $\SC(a) \le (2+\alpha) \cdot \SC(b)$. 
	
	Due to \Cref{cor:integral-domination}, the integral domination graph $G^{\elec}(a)$ admits a perfect matching $M : V \to V$ which satisfies $a \succeq_i \top(M(i))$ for each $i \in V$. Thus,
    \begin{align*}
        \SC(a) &= \sum_{i \in V} d(i,a)\\
        &\le \sum_{i \in V} d(i,\top(M(i))) &(\because a \succeq_i \top(M(i)), \forall i \in V)\\
        &\le \sum_{i \in V} \big( d(i,b) + d(b,\top(M(i))) \big) &(\because \text{triangle inequality})\\
        &= \SC(b) + \sum_{i \in V} d(b,\top(M(i)))\\
        &= \SC(b) + \sum_{i \in V} d(b,\top(i)) &(\because M \text{ is a perfect matching})\\
        &= \SC(b) + \sum_{i \in V : \top(i) \neq b} d(b,\top(i)) &(\because d(b,b)=0)\\
        &\le \SC(b) + \sum_{i \in V : \top(i) \neq b} \left(d(b,i)+d(i,\top(i))\right) &(\because \text{triangle inequality})\\
        &\le \SC(b) + \sum_{i \in V : \top(i) \neq b} \left(d(b,i)+\alpha \cdot d(i,b)\right) &(\because \text{$\alpha$-decisiveness})\\
        &\le \SC(b) + \sum_{i \in V} \left(d(b,i)+\alpha \cdot d(i,b)\right)\\
        &= (2+\alpha) \cdot \SC(b).
    \end{align*}

    Because the metric space, the election, and the choice of $b$ were arbitrary, this establishes that \plumatching has distortion at most $2+\alpha$.
    
    To show that this analysis is tight for $m \ge 3$, consider an election with two voters ($V = \set{1,2}$) and three candidates ($C = \set{a, b, c}$). The preference profile is as follows: $\sigma_1 = a \succ b \succ c$ and $\sigma_2 = c \succ b \succ a$. First, note that the integral domination graph of $b$ has two edges: $(1,2)$ and $(2,1)$. Both edges together form a perfect matching. Hence, \plumatching may return $b$.
    However, consider an $\alpha$-decisive metric consistent with the preference profile given by the following undirected graph, where the distance between any two points is the shortest distance in the graph. 
    
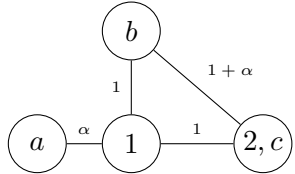
\[\begin{tikzpicture}
        \tikzstyle{main_node} = [circle,fill=white,draw,minimum size=2em,inner sep=0pt]

        \node[main_node] (a) at (0, 0) {$a$};
        \node[main_node] (i) at (1.25, 0) {$1$};
        \node[main_node] (jc) at (3, 0) {$2, c$};
        \node[main_node] (b) at (1.25, 1.5) {$b$};

      \draw (a) -- (i) node [midway, above] {\tiny $\alpha$};
      \draw (i) -- (jc) node [midway, above] {\tiny $1$};
      \draw (i) -- (b) node [midway, left] {\tiny $1$};
      \draw (b) -- (jc) node [midway, above right] {\tiny $1 + \alpha$};
    \end{tikzpicture}\]
    Note that $\SC(b) = 2+\alpha$, while $\SC(c) = 1$.\footnote{In this case, $c$ is also a candidate that \plumatching may return, leading one to hope that tie-breaking might help reduce the distortion. However, in Appendix~\ref{sec:ties}, we provide a more complex construction which shows that the lower bound effectively holds regardless of the tie-breaking.} Hence, the distortion of \plumatching is at least $2 + \alpha$. For $m > 3$, we can derive the same lower bound by placing the remaining candidates sufficiently far away.
\end{proof}

Note that in the trivial case of $m=2$, \plumatching (and every other reasonable voting rule) coincides with the majority rule, which picks the top choice of a majority of voters. This is known to have distortion $1+2\alpha$, which is optimal for deterministic rules for every $\alpha \in [0,1]$~\cite{anshelevich2017randomized}. 

\begin{corollary}\label{cor:det-3-dist}
	For every $m$, the distortion of \plumatching is $3$, which is optimal across all deterministic social choice rules. 
\end{corollary}

This resolves the open conjecture by establishing that the optimal distortion of deterministic rules is $3$, and can be achieved by a polynomial-time computable combinatorial rule. 

While our bound is tight for every $\alpha \in [0,1]$ when $m = 2$, and for $\alpha = 1$ when $m \ge 3$, one might wonder if it is also tight for $\alpha < 1$ when $m \ge 3$.
The best known lower bound for deterministic rules is $1+2\alpha$~\cite{anshelevich2017randomized}, which could be significantly lower than our upper bound of $2+\alpha$. We derive an improved lower bound for deterministic rules, which shows that at least as $m \to \infty$, our upper bound is tight for every $\alpha \in [0,1]$. Note that our lower bound coincides with the aforementioned $1 + 2\alpha$ lower bound
when $m$ is $2$ or $3$, but improves as $m$ increases.

\begin{theorem}\label{thm:det-lower-bound}
	For every $m \ge 2$ and $\alpha \in [0,1]$, the distortion of every deterministic social choice rule is at least $2 + \alpha - 2(1 - \alpha)/\floor{m}_{even}$ for $\alpha$-decisive metric spaces, where $\floor{m}_{even}$
    is the largest even integer smaller than or equal to $m$.
\end{theorem}
\begin{proof}
	We show the lower bound of $2+\alpha-2(1-\alpha)/m$ for every even $m$; for odd $m$, we can simply add an irrelevant candidate that is sufficiently far away from all other candidates and voters. 
	
	Assume that $m = 2 \ell$ for some $\ell \in \bbN$. Consider an election $\elec = (V,C,\vsigma)$, where $V = \set{1,\ldots,2\ell}$, $C = \set{a_1,\ldots,a_{2\ell}}$, and the preference profile $\vsigma$ is as follows. 
	Define $A = \set{a_1,\ldots,a_{\ell}}$, and $B = C \setminus A$, $V_A = \set{1,\ldots,\ell}$, and $V_B = V \setminus V_A$. For $i \in V_A$, voter $i$ ranks candidate $a_i \in A$ first, followed by the remaining candidates in $A$ in an arbitrary order, followed by the candidates in $B$ in an arbitrary order. For $i \in V_B$, voter $i$ ranks candidate $a_i \in B$ first, followed by the remaining candidates in $B$ in an arbitrary order, followed by the candidates in $A$ in an arbitrary order.

    Let $f$ be a deterministic social choice rule. Without loss of generality, suppose $f$ returns a candidate from $A$; the argument is symmetric when $f$ returns a candidate from $B$. We now construct an $\alpha$-decisive metric space consistent with $\vsigma$ under which every candidate in $A$ has an approximation ratio to social cost at least $2 + \alpha - 2 (1-\alpha)/m$. The metric is represented by the graph in \Cref{fig:lower-bound-metric}
    \begin{figure}
        \begin{center}
            \begin{tikzpicture}[scale=1.5,every text node part/.style={align=center}]
                \tikzstyle{main_node} = [circle,fill=white,draw,minimum size=2.5em,inner sep=3pt]

                \node[main_node] (a1) at (0,0) {$a_1$};
                \node[main_node] (a2) at (0, -1) {$a_2$};
                \node at (0, -1.9) {\vdots};
                \node[main_node] (al) at (0, -3) {$a_\ell$};

                \node[main_node] (va1) at (2,0) {$1$};
                \node[main_node] (va2) at (2, -1) {$2$};
                \node at (2, -1.9) {\vdots};
                \node[main_node] (val) at (2, -3) {$\ell$};

                \node[main_node] (p) at (4, -1.5) {$a_{\ell+1},\ldots,a_{2\ell}$\\[0.2cm]$\ell+1,\ldots,2\ell$};

                \draw (p) -- (va1) node [midway, above] {\tiny $1$};
                \draw (p) -- (va2) node [midway, above] {\tiny $1$};
                \draw (p) -- (val) node [midway, above] {\tiny $1$};

                \draw (a1) -- (va1) node [midway, above] {\tiny $\alpha$};
                \draw (a2) -- (va2) node [midway, above] {\tiny $\alpha$};
                \draw (al) -- (val) node [midway, above] {\tiny $\alpha$};

                \draw (a1) -- (va2) node [pos=0.15, above] {\tiny $1$};
                \draw (a1) -- (val) node [pos=0.15, above] {\tiny $1$};

                \draw (a2) -- (va1) node [pos=0.85, above] {\tiny $1$};
                \draw (a2) -- (val) node [pos=0.85, below] {\tiny $1$};

                \draw (al) -- (va1) node [pos=0.85, above] {\tiny $1$};
                \draw (al) -- (va2) node [pos=0.15, below] {\tiny $1$};
            \end{tikzpicture}
        \end{center}
	\caption{The metric space construction used in the proof of \Cref{thm:det-lower-bound}.}
	\label{fig:lower-bound-metric}
	\end{figure}
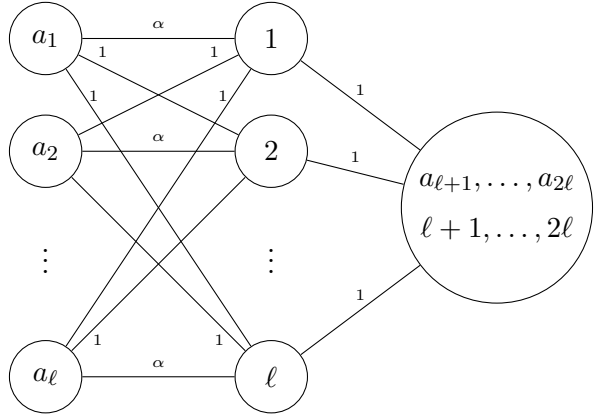
    
    For $i \in V_A$, voter $i$ has distance $\alpha$ from $a_i$, and $1$ from every other candidate. Voters in $V_B$ and candidates in $B$ are located at a single point at distance $1$ from each voter in $V_A$. The distance between any other pair of points is the shortest distance between them in the resulting graph. 
    
    It is easy to check that this metric is $\alpha$-decisive and consistent with $\vsigma$.
	Further, we have that for each $a_i \in A$, $\SC(a_i) = 1 \cdot \alpha + (\ell-1) \cdot 1 + \ell \cdot (1+\alpha)$, whereas for each $a_i \in B$, $\SC(a_i) = \ell \cdot 1$. Since $f$ returns a candidate from $A$, its distortion is at least 
	\[
	\frac{\alpha + \ell-1+\ell \cdot (1+\alpha)}{\ell} = 2+\alpha-\frac{1-\alpha}{\ell} = 2+\alpha-\frac{2 \cdot (1-\alpha)}{m},
	\]
	as desired.
\end{proof}

Note that for even $m$, in the construction used in the proof of \Cref{thm:det-lower-bound} each voter ranks a distinct candidate first. Hence, the lower bound applies even in the special case of peer selection, in which the set of voters is identical to the set of candidates (and thus $\alpha = 0$). 

For odd $m$, the bound can be slightly improved by having more voters whose top choices are slightly unequally divided between two subgroups of candidates of sizes $\floor{m/2}$ and $\ceil{m/2}$. However, the improvement is marginal, so we avoid it for ease of presentation. 

\subsection{Connecting \plumatching to The Matching Uncovered Set}\label{subsec:connection-to-the-mus}

As mentioned in \Cref{sec:intro}, we devised the notion of integral domination graphs by directly building on the work of \citet{munagala2019improved}. They combinatorially defined a set of candidates that they named the matching uncovered set, and proved that every candidate in this set has distortion at most $3$, but failed to answer whether this set is always non-empty. To understand how this is connected to our work, let us begin by introducing their definitions.

\begin{definition}
	Given an election $\elec$, the \emph{separation graph} of candidates $a$ and $b$ is defined as the bipartite graph $G^\elec(a,b) = (V, V, E_{a,b})$ where $(i, j) \in E_{a,b}$ if there exists a candidate $c \in C$ such that $a \succeq_i c$ and $c \succeq_j b$. When the election is clear from the context, we will drop it from the notation.
\end{definition}

\begin{definition}
	The \emph{matching uncovered set} of an election $\elec$ is the set of candidates $a \in C$ such that for every $b \in C$, $G^{\elec}(a,b)$ admits a perfect matching.
\end{definition}

Note that the separation graph $G(a,b)$ has an edge from voter $i$ on the left to voter $j$ on the right if $a$ weakly defeats some candidate in $i$'s vote who weakly defeats $b$ in $j$'s vote. In comparison, recall that the integral domination graph $G(a)$ has an edge from $i$ to $j$ if $a$ weakly defeats $\top(j)$ in $i$'s vote. Because $\top(j)$ clearly weakly defeats every candidate in $j$'s vote, an $i \to j$ edge in $G(a)$ also exists in $G(a,b)$ for every candidate $b$. That is, for each $a$, the integral domination graph $G(a)$ is a subgraph of the separation graph $G(a,b)$ for \emph{every} $b$. Thus, seeking a perfect matching in $G(a)$, as \plumatching does, is equivalent to seeking a single set of edges that form a perfect matching in $G(a,b)$ for every $b$. This is stronger than the requirement imposed by the matching uncovered set, which is simply that each $G(a,b)$ admits a (potentially different) perfect matching. We therefore have the following. 

\begin{proposition}\label{prop:plumatching-matchunc}
    \plumatching always returns a candidate in the matching uncovered set. Consequently, the matching uncovered set is always non-empty.
\end{proposition}

When $\alpha=1$, \Cref{prop:plumatching-matchunc}, along with the proof by \citet{munagala2019improved} that every candidate in the matching uncovered set has distortion at most $3$, is sufficient to establish that \plumatching has distortion at most $3$. 

The benefits of our alternative analysis of \plumatching in \Cref{thm:det-upper-bound} are threefold.
First, our derivation of the upper bound on distortion is much simpler than the tedious derivation by \citet{munagala2019improved}. Second, as we show in Appendix~\ref{sec:fairness}, our analysis seamlessly extends to produce an upper bound on the so-called \emph{fairness ratio}~\cite{goel2017metric}, which is not known to be possible with the previous argument.
More importantly, our analysis yields an upper bound of $2+\alpha$ on distortion, which, for $\alpha < 1$, is stronger than the upper bound of $3$ established by \citet{munagala2019improved}. However, one might wonder whether the stronger guarantees of the Ranking-Matching Lemma are not really necessary to achieve this improved bound. Could it be the case that selecting any candidate from the matching uncovered set already achieves distortion at most $2+\alpha$, even when $\alpha < 1$? We answer this negatively by providing two different lower bounds in \Cref{prop:matchunc-bad,prop:condorcet-bad}. 

\begin{proposition}\label{prop:matchunc-bad}
    For every $m \ge 5$ and $\alpha \in [0,1]$, selecting an arbitrary candidate from the matching uncovered set yields distortion at least $2.5 + \alpha/2$ for $\alpha$-decisive metric spaces. 
\end{proposition}
\begin{proof}
	Below, we produce an election with $m=5$ candidates in which some candidate from the matching uncovered set yields approximation ratio at least $2.5+\alpha/2$. For $m > 5$, we can simply add dummy candidates that are sufficiently far from all the voters. 
	
	Consider an election $\elec = (V,C,\vsigma)$ with $5$ candidates ($C = \set{a, b, c, d, e}$) and $3$ voters ($V = \{1, 2, 3\}$) having the following preferences:
    \[\begin{array}{ll}
          \sigma_1: & b \succ e \succ c \succ a \succ d\\
          \sigma_2: & c \succ d \succ b \succ a \succ e\\
          \sigma_3: & d \succ a \succ c \succ b \succ e
    \end{array}\]
    First, we argue that $a$ is in the matching uncovered set. \Cref{fig:seperation-graphs} shows separation graphs $G(a,b)$, $G(a,c)$, $G(a,d)$, and $G(a,e)$, with a perfect matching in each highlighted with thick red edges. 
    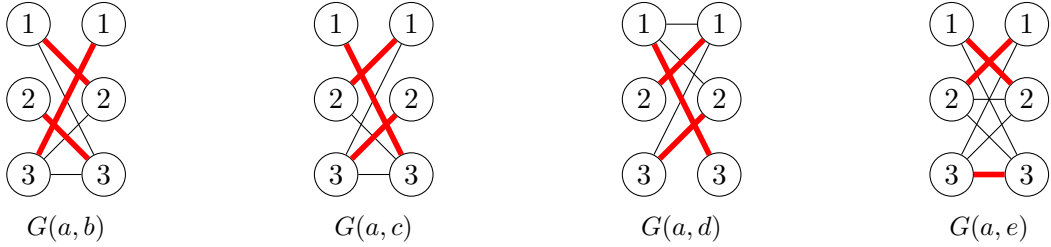
\begin{figure}
        \captionsetup[subfigure]{labelformat=empty}
        \begin{subfigure}[b]{0.24\textwidth}
            \centering
            \begin{tikzpicture}
                \tikzstyle{main_node} = [circle,fill=white,draw,minimum size=1.5em,inner sep=0pt]

                \node[main_node] (il) at (0, 0) {$1$};
                \node[main_node] (jl) at (0, -1) {$2$};
                \node[main_node] (kl) at (0, -2) {$3$};
                \node[main_node] (ir) at (1, 0) {$1$};
                \node[main_node] (jr) at (1, -1) {$2$};
                \node[main_node] (kr) at (1, -2) {$3$};

                \draw (il)--(kr);
                \draw (kl)--(jr);
                \draw (kl)--(kr);
                \draw (il)--(jr) [draw=red, line width=.75mm];
                \draw (jl)--(kr) [draw=red, line width=.75mm];
                \draw (kl)--(ir) [draw=red, line width=.75mm];
            \end{tikzpicture}
            \caption{$G(a, b)$}
        \end{subfigure}
        \begin{subfigure}[b]{0.24\textwidth}
            \centering
            \begin{tikzpicture}
                \tikzstyle{main_node} = [circle,fill=white,draw,minimum size=1.5em,inner sep=0pt]

                \node[main_node] (il) at (0, 0) {$1$};
                \node[main_node] (jl) at (0, -1) {$2$};
                \node[main_node] (kl) at (0, -2) {$3$};
                \node[main_node] (ir) at (1, 0) {$1$};
                \node[main_node] (jr) at (1, -1) {$2$};
                \node[main_node] (kr) at (1, -2) {$3$};

                \draw (jl)--(kr);
                \draw (kl)--(ir);
                \draw (kl)--(kr);
                \draw (il)--(kr) [draw=red, line width=.75mm];
                \draw (jl)--(ir) [draw=red, line width=.75mm];
                \draw (kl)--(jr) [draw=red, line width=.75mm];
            \end{tikzpicture}
            \caption{$G(a, c)$}
        \end{subfigure}
        \begin{subfigure}[b]{0.24\textwidth}
            \centering
            \begin{tikzpicture}
                \tikzstyle{main_node} = [circle,fill=white,draw,minimum size=1.5em,inner sep=0pt]

                \node[main_node] (il) at (0, 0) {$1$};
                \node[main_node] (jl) at (0, -1) {$2$};
                \node[main_node] (kl) at (0, -2) {$3$};
                \node[main_node] (ir) at (1, 0) {$1$};
                \node[main_node] (jr) at (1, -1) {$2$};
                \node[main_node] (kr) at (1, -2) {$3$};

                \draw (il)--(ir);
                \draw (il)--(jr);
                \draw (kl)--(ir);
                \draw (il)--(kr) [draw=red, line width=.75mm];
                \draw (jl)--(ir) [draw=red, line width=.75mm];
                \draw (kl)--(jr) [draw=red, line width=.75mm];
            \end{tikzpicture}
            \caption{$G(a, d)$}
        \end{subfigure}
        \begin{subfigure}[b]{0.24\textwidth}
            \centering
            \begin{tikzpicture}
                \tikzstyle{main_node} = [circle,fill=white,draw,minimum size=1.5em,inner sep=0pt]

                \node[main_node] (il) at (0, 0) {$1$};
                \node[main_node] (jl) at (0, -1) {$2$};
                \node[main_node] (kl) at (0, -2) {$3$};
                \node[main_node] (ir) at (1, 0) {$1$};
                \node[main_node] (jr) at (1, -1) {$2$};
                \node[main_node] (kr) at (1, -2) {$3$};

                \draw (il)--(kr);
                \draw (jl)--(jr);
                \draw (jl)--(kr);
                \draw (kl)--(ir);
                \draw (kl)--(jr);
                \draw (il)--(jr) [draw=red, line width=.75mm];
                \draw (jl)--(ir) [draw=red, line width=.75mm];
                \draw (kl)--(kr) [draw=red, line width=.75mm];
            \end{tikzpicture}
            \caption{$G(a, e)$}
        \end{subfigure}
        \caption{Separation graphs in the proof of \Cref{prop:matchunc-bad}. Thick red edges show perfect matchings.}
        \label{fig:seperation-graphs}
    \end{figure}

	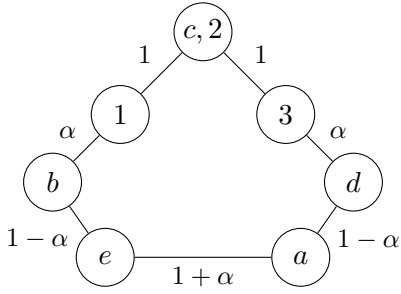
\begin{figure}
        \centering
        \begin{tikzpicture}
            \tikzstyle{main_node} = [circle,fill=white,draw,minimum size=2em,inner sep=0pt]

            \node[main_node] (cj) at (0, 0) {$c, 2$};
            \node[main_node] (i) at (-1.1, -1.1) {$1$};
            \node[main_node] (k) at (1.1, -1.1) {$3$};
            \node[main_node] (b) at (-2, -2) {$b$};
            \node[main_node] (d) at (2, -2) {$d$};
            \node[main_node] (e) at (-1.3, -3) {$e$};
            \node[main_node] (a) at (1.3, -3) {$a$};

            \draw (cj)--(k) node [midway, above right] {\small $1$};
            \draw (k)--(d) node [midway, above right] {\small $\alpha$};
            \draw (d)--(a) node [midway, below right] {\small $1 - \alpha$};
            \draw (a)--(e) node [midway, below] {\small $1 + \alpha$};
            \draw (e)--(b) node [midway, below left] {\small $1 - \alpha$};
            \draw (b)--(i) node [midway, above left] {\small $\alpha$};
            \draw (i)--(cj) node [midway, above left] {\small $1$};
        \end{tikzpicture}
        \caption{An $\alpha$-decisive metric space used in the proof of \Cref{prop:matchunc-bad}.}
        \label{fig:kamesh-lower-metric}
    \end{figure}

	However, consider the metric space shown in \Cref{fig:kamesh-lower-metric}; as usual, the distance between any pair of points is the shortest distance in the graph shown. It can be checked easily that this metric is $\alpha$-decisive and consistent with $\vsigma$. Note that under this metric, $\SC(a) = 5+\alpha$, whereas $\SC(c) = 2$. Thus, a rule picking $a$ would have distortion at least $(5+\alpha)/2 = 2.5+\alpha/2$. 
	
	A curious reader may check that in this example, the integral domination graph of $a$ does not admit a perfect matching, while those of $c$ and $d$ do. Thus, \plumatching may pick either of $c$ or $d$. While $c$ is optimal, $d$ has approximation ratio $2+\alpha$, matching our upper bound.
\end{proof}

In fact, for large $m$, we can prove a better lower bound on the distortion that can be achieved by picking an arbitrary candidate from the matching uncovered set, or even by picking a special candidate from this set that has received significant attention in the voting literature.

\begin{definition}
	Given an election $\elec = (V,C,\vsigma)$, we say that candidate $a$ is a \emph{Condorcet winner} if for every other candidate $b$, there are at least $n/2$ voters (a majority) who prefer $a$ to $b$. A Condorcet winner is not guaranteed to exist. A deterministic social choice rule is called \emph{Condorcet consistent} if it outputs a Condorcet winner on every election in which one exists. 
\end{definition}

It is easy to argue that every Condorcet winner $a$ is in the matching uncovered set.\footnote{This is because for every other candidate $b$, a perfect matching in $G(a,b)$ can be obtained easily by matching the (at most $n/2$) voters on the right who prefer $b$ over $a$ to (some of at least $n/2$) voters on the left who prefer $a$ over $b$, and matching the remaining voters on the right to the remaining voters on the left arbitrarily.} Thus, any lower bound on distortion that applies to social choice rules that pick a Condorcet winner when one exists (i.e. to Condorcet consistent rules) also applies to picking an arbitrary candidate from the matching uncovered set. The following result establishes such a lower bound.

\begin{proposition}\label{prop:condorcet-bad}
	For every $m \ge 2$ and $\alpha \in [0,1]$, the distortion of every deterministic Condorcet consistent rule is at least $3-\frac{4(3+\alpha)}{m+2+m\alpha}$ for $\alpha$-decisive metric spaces. The same lower bound applies to selecting an arbitrary candidate from the matching uncovered set. 
\end{proposition}
\begin{proof}
	Let $k = m-2$. Consider an election $\elec = (V,C,\vsigma)$, where the set of candidates is $C = \set{a,b,c_1,\ldots,c_k}$, and the set of voters is $V = \cup_{c \in C} V_c$ with $|V_a| = 2$, $|V_b| = k$, and $|V_c| = 1$ for each $c \in C\setminus\set{a,b}$. The preference profile $\vsigma$ is as follows. 
	\[\begin{array}{lll}
	\forall i \in V_a, &\sigma_i: & a \succ c_1 \succ \ldots \succ c_k \succ b\\
	\forall i \in V_b, &\sigma_i: & b \succ c_1 \succ \ldots \succ c_k \succ a\\
	\forall r \in [k], i \in V_{c_r}, &\sigma_i: & c_r \succ a \succ b \succ c_1 \succ \ldots \succ c_{r-1} \succ c_{r+1} \succ \ldots \succ c_k
	\end{array}\]	
	
	It is easy to check that in this preference profile, $a$ is the unique Condorcet winner. However, consider the metric space shown in \Cref{fig:condorcet-bad}. As usual, all pairwise distances are the corresponding shortest distances in the graph. 
	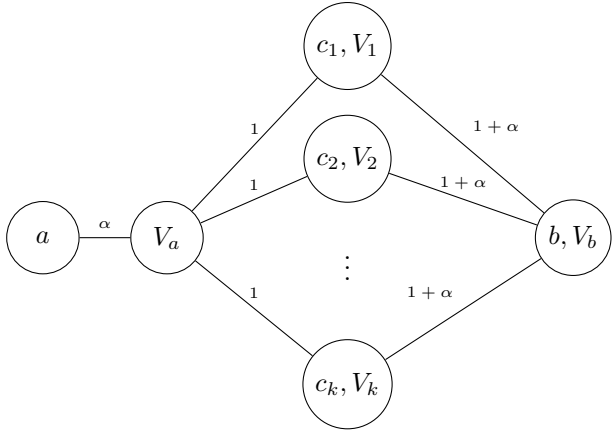
\begin{figure}
		\centering
		\begin{tikzpicture}[scale=1.5]
		\tikzstyle{main_node} = [circle,fill=white,draw,minimum size=2.5em,inner sep=3pt]
		
		\node[main_node] (a) at (-2.7, 0) {\small$a$};
		\node[main_node] (va) at (-1.6, 0) {\small$V_a$};
		\node[main_node] (c1) at (0, 1.7) {\small$c_1, V_1$};
		\node[main_node] (c2) at (0, .7) {\small$c_2, V_2$};
		\node[main_node] (ck) at (0, -1.3) {\small$c_k, V_k$};
		\node at (0, -.2) {\small$\vdots$};
		\node[main_node] (b) at (2, 0) {\small$b, V_b$};

		\draw (a) -- (va) node [midway, above] {\tiny $\alpha$};
		\draw (va) -- (c1) node [midway, above] {\tiny $1$};
		\draw (va) -- (c2) node [midway, above] {\tiny $1$};
		\draw (va) -- (ck) node [midway, above] {\tiny $1$};
		\draw (c1) -- (b) node [midway, above right] {\tiny $1 + \alpha$};
		\draw (c2) -- (b) node [midway, above] {\tiny $1 + \alpha$};
		\draw (ck) -- (b) node [midway, above left] {\tiny $1 + \alpha$};
		\end{tikzpicture}
		\caption{An $\alpha$-decisive metric space used in the proof of \Cref{prop:condorcet-bad}.}
		\label{fig:condorcet-bad}
	\end{figure}

	One can check that this is an $\alpha$-decisive metric space consistent with $\vsigma$. However, $\SC(a) = 2 \alpha + k \cdot (1+\alpha) + k \cdot (2+2\alpha) = 3k+(3k+2)\alpha$, whereas $\SC(b) = k \cdot (1+\alpha) + 2 \cdot (2+\alpha) = k+4+(k+2)\alpha$. Hence, the distortion of a rule that selects $a$ on this preference profile is at least
	\[
	\frac{3k+(3k+2)\alpha}{k+4+(k+2)\alpha} = 3-\frac{4(3+\alpha)}{m+2+m\alpha},
	\]
	where the last step uses the fact that $k=m-2$.
\end{proof}

\citet{anshelevich2015approximating} proved that picking a Condorcet winner, when one exists, guarantees $3$-approximation to social cost. Thus, if we were only interested in bounding the distortion for general metric spaces ($\alpha=1$), then being Condorcet consistent would not hurt: indeed, any social choice rule can be made Condorcet consistent by forcing it to pick a Condorcet winner on profiles where one exists without increasing its distortion.  

However, if our goal was to also achieve improved distortion bounds parametrized by $\alpha$, \Cref{prop:condorcet-bad} shows that being Condorcet consistent is at odds with this goal. When $m \to \infty$, the lower bound in \Cref{prop:condorcet-bad} converges to $3$ regardless of the value of $\alpha$. In contrast, recall that \plumatching achieves distortion $2+\alpha$ for any $m$. Thus, while \plumatching is not Condorcet consistent,\footnote{In Appendix~\ref{sec:condorcet}, we present a simple necessary condition for a candidate to be possibly returned by \plumatching, and a simpler example than the one used in the proof of \Cref{prop:condorcet-bad} where a Condorcet winner violates this condition.} this is necessary for it to achieve the improved distortion. Similarly, \Cref{prop:matchunc-bad,prop:condorcet-bad} also show that picking an arbitrary candidate from the matching uncovered set yields a worse parametrized distortion bound than that of \plumatching.

\section{The Distortion of Randomized Social Choice Rules}\label{sec:rand-dist}
\newcommand{\agps}{$\alpha$-\textsc{Generalized Proportional to Squares}\xspace}

In the previous section, we settled the question of determining the optimal distortion of deterministic social choice rules, and showed that it is $3$ for all $m \ge 2$. The obvious next step is to study the power of randomization in the metric distortion framework, and investigate whether it allows achieving any lower distortion. In this section, we show that randomization indeed helps, but leave open the question of determining the optimal distortion of randomized rules. 

Randomized rules were first considered in this framework by \citet{anshelevich2017randomized}. They studied a natural rule called \randdict, which chooses one of the voters uniformly at random, and outputs the top candidate of that voter; in other words, each candidate $a$ is chosen with probability equal to $\plu(a)/\sum_{c\in C} \plu(c)$, or $\plu(a)/n$. They proved that its distortion is $2+\alpha$.\footnote{Specifically, they proved a distortion bound of $2+\alpha-2/n$ for $n$ voters. But since we defined distortion using supremum over preference profiles of all sizes, it becomes $\sup_{n \ge 1} 2+\alpha-2/n = 2+\alpha$ in our framework. That said, the distortion of \randdict is strictly less than $3$ for any finite $n$, which is impossible for deterministic rules.} Given our result (\Cref{thm:det-upper-bound}) that \plumatching achieves distortion $2+\alpha$ deterministically, this raises the question of whether randomization helps. 

For general metric spaces, i.e., when $\alpha=1$, \citet{kempe2020communication} answered this question by proposing a different randomized voting rule which achieves a slightly better distortion of $3-2/m$. This rule, like \randdict, uses only the plurality scores of the candidates, and it is known that $3-2/m$ is the best possible distortion that any rule can achieve for $\alpha=1$ using only the plurality scores~\cite{gross2017vote,goel2017metric}. 

\emph{What about the case of $\alpha < 1$?} \citet{anshelevich2017randomized} proved that a randomized social choice rule, even using the full preference rankings, cannot achieve distortion better than $1+\alpha$ when $m \ge 2$. Interestingly, this bound is not met by \randdict even for the seemingly trivial case of $m=2$ candidates, where each vote is one of only two possible rankings. To bridge this gap, specifically for the case of $m=2$ candidates (say $a$ and $b$), they proposed a rule, \agps, which achieves the optimal distortion of $1+\alpha$. This rule chooses alternative $a$ with probability
\begin{equation}\label{eqn:agps}
	\frac{(1+\alpha) \cdot \plu(a)^2 - (1-\alpha) \cdot \plu(a) \cdot \plu(b)}{(1+\alpha) \cdot \left(\plu(a)^2 + \plu(b)^2\right) - 2(1-\alpha) \cdot \plu(a) \cdot \plu(b)},
\end{equation}
and $b$ with the remaining probability. 

Our main contribution in this section is twofold. We propose a novel randomized social choice rule, \smartdict, which achieves distortion at most $2+\alpha-2/m$ (\Cref{thm:smartdict}). This matches the distortion of $3-2/m$ achieved by \citet{kempe2020communication} for $\alpha=1$, but improves when $\alpha < 1$. \smartdict, like \randdict, \agps, and the rule proposed by \citet{kempe2020communication}, uses only the plurality scores. In fact, for $m=2$, it coincides with \agps and achieves the optimal distortion of $1+\alpha$. Further, we show that no randomized social choice rule which uses only the plurality scores can achieve distortion lower than $2+\alpha-2/m$ (\Cref{thm:smartdict-opt-plurality-scores}). Thus, \smartdict achieves the optimal distortion achievable using only plurality scores \emph{for every $\alpha$}. We also improve upon the $1+\alpha$ lower bound of \citet{anshelevich2017randomized} that applies to rules which are allowed to use the full preference rankings (\Cref{thm:randomized-lower-bound}).

Recall that in \Cref{thm:det-lower-bound}, we showed that deterministic rules cannot achieve distortion better than $2+\alpha-2(1-\alpha)/m$ for even $m$. Note that $2+\alpha-2/m$, which is the distortion achieved by \smartdict, is strictly smaller when $\alpha > 0$. Hence, this establishes a strict separation between deterministic and randomized rules for every $\alpha > 0$ (which was previously only established for $\alpha=1$), but leaves open the question when $\alpha = 0$.

\subsection{Randomized \smartdict Rule}

We begin by defining \smartdict.

\begin{definition}
Let $\elec = (V,C,\vsigma)$ be an election. If there exists a candidate $a$ with $\plu(a)\geq (1+\alpha) \cdot \frac{n}{2}$, \smartdict arbitrarily chooses one such candidate with probability $1$. Otherwise, it chooses each candidate $a$ with probability proportional to $\plu(a)/\left(n-\frac{2}{1+\alpha} \cdot  \plu(a)\right)$, i.e., with probability equal to
\begin{equation}\label{eq:prob-def}
\frac{\plu(a)}{\left(n-\frac{2}{1+\alpha} \cdot  \plu(a)\right)\sum_{c\in C} \frac{\plu(c)}{n-\frac{2}{1+\alpha} \cdot  \plu(c)}}.
\end{equation}
\end{definition}

We leave it as an exercise to the reader to verify that the expression in \Cref{eq:prob-def} coincides with the expression in \Cref{eqn:agps} when $C = \set{a,b}$. We also note that for $\alpha=1$, \smartdict simply selects each candidate $a$ with probability proportional to $\plu(a)/(n-\plu(a)$. While the distortion of \smartdict for $\alpha=1$, proved in \Cref{thm:smartdict}, matches the distortion of the rule proposed by \citet{kempe2020communication}, the two rules are different when $m \ge 3$.

In order to prove the distortion bound for \smartdict, we make use of the following results by \citet{anshelevich2017randomized}.
\begin{lemma}\label{lem:known-dist-bound}[\citealp{anshelevich2017randomized}]
Let $\elec = (V,C,\vsigma)$ be an election induced by an $\alpha$-decisive metric space with distance metric $d$. Let $c^*$ be an optimal candidate. Then, the approximation ratio of a randomized social choice rule choosing each candidate $a$ with probability $\Pr[a]$ on this instance is at most
\[
1+\frac{(1+\alpha)\sum_{a\in C} \Pr[a] \cdot \left(n-\frac{2}{1+\alpha}\cdot  \plu(a)\right)\cdot d(c,c^*)}{\sum_{a\in C} \plu(a) \cdot d(a,c^*)}.
\]
\end{lemma}

As an intermediate step toward proving this lemma, they also showed the following.

\begin{proposition}\label{prop:known-inequality}[\citealp{anshelevich2017randomized}]
Let $\elec = (V,C,\vsigma)$ be an election induced by an $\alpha$-decisive metric space with distance metric $d$. Let $c^*$ be an optimal candidate. Then, for every candidate $a$,
\[
\SC(a) \leq \SC(c^*)+\left(n-\frac{2}{1+\alpha} \cdot \plu(a)\right) \cdot d(a,c^*).
\]
\end{proposition}

We also prove the following lemma that will help us simplify our proof of the distortion bound. 

\begin{lemma}\label{lem:summation-bound}
For any $w \in (0,1]$, $n\in \mathbb{N}$, and $(x_i)_{i=1}^m\in [0, n/w)^m$ such that $\sum_{i=1}^m x_i = n$,
\[\sum_{i=1}^m \frac{x_i}{n-w \cdot x_i} ~\geq~ \frac{m}{m-w}. \]
\end{lemma}
\begin{proof}
We will show that the summation on the left hand side of the lemma's statement is minimized when $x_i=n/m$ for all $i\in \set{1,\dots,m}$, which will imply that 
\[
\sum_{i=1}^m \frac{x_i}{n-w \cdot x_i} ~\geq~ \sum_{i=1}^m \frac{n/m}{n-w \cdot n/m} ~=~ \sum_{i=1}^m \frac{1}{m-w} ~=~ \frac{m}{m-w}.
\]

Let $\mathcal{X}$ denote the set of feasible vectors $x = (x_i)_{i=1}^m$ that minimize the summation on the left. Aiming for a contradiction, assume that for every vector $x \in \mathcal{X}$, there exist two values, $x_j$ and $x_k$, such that $x_j\neq x_k$. Let $x^*$ be a vector in $\mathcal{X}$ that is leximin-optimal across all vectors in $\mathcal{X}$. That is, we choose $x^*$ from $\mathcal{X}$ by maximizing the smallest $x_i$, then breaking ties in favor of higher second-smallest $x_i$, then in favor of higher third-smallest $x_i$, and so on. In other words, we sort each vector in a non-decreasing order, and then compare them lexicographically, picking a maximum.

Take any two values $x^*_j$ and $x^*_k$ such that $x^*_j> x^*_k$, and consider their contribution to the sum $\sum_{i=1}^m \frac{x^*_i}{n-w x^*_i}$, which is equal to
\begin{equation}\label{eq:contribution}
\frac{x^*_j}{n-w x^*_j} + \frac{x^*_k}{n-w x^*_k} 
~=~ \frac{(n-w x^*_k)x^*_j+(n-w \cdot x^*_j)x^*_k}{(n-w x^*_j)(n-w x^*_k)}
~=~ \frac{n (x^*_j+x^*_k)-2 w x^*_k x^*_j}{n^2-n(x^*_j+x^*_k)+w^2 x^*_j x^*_k}.
\end{equation}
Now, consider an alternative vector $x' = (x'_i)_{i=1}^m$ such that $x'_j=x'_k=(x^*_j+x^*_k)/2$ and $x'_i = x^*_i$ for all $i\neq j,k$. Clearly, these values also satisfy $\sum_{i=1}^m x'_i = n$ since this was satisfied by $x^*$, only values of $j$ and $k$ changed, and $x'_j + x'_k = x^*_j+x^*_k$. Also, since $x^* \in [0,n/w)^m$ and $x^*_k < x'_j = x'_k < x^*_j$, we have that $x' \in [0,n/w)^m$ (thus $x'$ is a feasible vector) and that $x'$ is strictly better in the leximin-ordering than $x^*$. 

Finally, due to the arithmetic mean-geometric mean inequality, 
\[
\sqrt{x^*_j x^*_k} \leq \frac{x^*_j + x^*_k}{2} ~\Rightarrow~  x^*_j x^*_k \leq \left(\frac{x^*_j + x^*_k}{2}\right)^2 = x'_j x'_k.
\]
Therefore, $x^*_j+x^*_k = x'_j + x'_k$ and $x^*_j x^*_k \leq x'_j  x'_k$. By observing the right hand side of \Cref{eq:contribution}, and noting the fact that the numerator is decreasing in $x^*_j x^*_k$ and the denominator is decreasing in $x^*_j x^*_k$, we conclude that
\begin{align*}
&\phantom{\Rightarrow} \frac{n (x^*_j+x^*_k)-2 w x^*_k x^*_j}{n^2-n(x^*_j+x^*_k)+w^2 x^*_j x^*_k} ~\leq~ \frac{n (x'_j+x'_k)-2 w x'_k x'_j}{n^2-n(x'_j+x'_k)+w^2 x'_j x'_k} \\
&\Rightarrow \frac{x^*_j}{n-w x^*_j} + \frac{x^*_k}{n-w x^*_k} ~\leq~ \frac{x'_j}{n-w x'_j} + \frac{x'_k}{n-w x'_k} \\ 
&\Rightarrow \sum_{i=1}^m \frac{x^*_i}{n-w x^*_i} ~\leq~ \sum_{i=1}^m \frac{x'_i}{n-w x'_i}.
\end{align*}

Hence, $x'$ is a feasible vector minimizing the summation, and is therefore in $\mathcal{X}$. Further, it is strictly better than $x^*$ in the leximin ordering, which contradicts the fact that $x^*$ is leximin-optimal among vectors in $\mathcal{X}$. 
\end{proof}

Using these results, we are now ready to prove the improved distortion bound that we can achieve using \smartdict. 

\begin{theorem}\label{thm:smartdict}
	For every $m \ge 2$ and $\alpha \in [0,1]$, \smartdict has distortion at most $2+\alpha-2/m$ for $\alpha$-decisive metric spaces. 
\end{theorem}
\begin{proof}
Let $\elec = (V,C,\vsigma)$ be an election. If there exists any candidate $a$ with $\plu(a) \geq (1+\alpha)\cdot \frac{n}{2}$, note that \smartdict arbitrarily chooses one such candidate with probability 1. However, using \Cref{prop:known-inequality}, we can conclude that for this candidate $c$, we have $\SC(c) \leq \SC(c^*)$, where $c^*$ is an optimal candidate. Hence, any such candidate is optimal as well, yielding an approximation ratio of $1$. 

If this is not the case, then we substitute the probability $\Pr[a]$ with which \smartdict chooses each candidate $a$ into \Cref{lem:known-dist-bound}, and get that the approximation ratio is at most
\begin{align*}
1+\frac{(1+\alpha)\sum_{a\in C}\plu(a) d(a,c^*)}{\left(\sum_{c\in C} \frac{\plu(c)}{n-\frac{2}{1+\alpha} \cdot \plu(c)}\right) \cdot \left(\sum_{a\in C}\plu(a) d(a,c^*)\right)} = 1+\frac{1+\alpha}{\sum_{c\in C} \frac{\plu(c)}{n-\frac{2}{1+\alpha} \cdot \plu(c)}}.
\end{align*}
Using \Cref{lem:summation-bound} with $w=\frac{2}{1+\alpha}$ and using the fact that $\plu(c) \in [0,(1+\alpha)\cdot\frac{n}{2})$ for each $c\in C$ and $\sum_{c\in C} \plu(c)=n$, we get that $\sum_{c\in C} \frac{\plu(c)}{n-\frac{2}{1+\alpha} \cdot  \plu(c)} \geq \frac{m}{m-\frac{2}{1+\alpha}}$. We therefore conclude that
\begin{align*}
\dist(\smartdict) ~\leq~ 1+(1+\alpha)\frac{m-\frac{2}{1+\alpha}}{m} ~=~ 2+\alpha - \frac{2}{m},
\end{align*}
as desired.
\end{proof}

A few remarks are in order. First, for $\alpha = 1$ (i.e.\ general metric spaces), the bound is $3-2/m < 3$, which shows a separation between randomized and deterministic rules for every $m \ge 2$. This is despite \smartdict only accessing plurality votes, while the deterministic rules being allowed to access the full ranked preference profile.\footnote{As mentioned earlier, this was also established by \citet{kempe2020communication}.} Second, as mentioned earlier, even as a function of $\alpha$, the upper bound of $2+\alpha-2/m$ for \smartdict is still better than the lower bound of $2+\alpha-2(1-\alpha)/m$ for deterministic rules when $m > 3$ is even and $\alpha > 0$. However, whether randomized rules can truly outperform deterministic rules when $\alpha = 0$, or when $\alpha = 1$ but $m \to \infty$, is still open.

\subsection{Lower Bounds}
Note that both \smartdict and \randdict access only the plurality scores of the candidates, in contrast to the deterministic \plumatching from the previous section which requires knowledge of the full preference profile. We first show that among randomized rules which only take the plurality scores as input, \smartdict is in fact optimal, even as a function of $\alpha$. This parametrizes the lower bound of $3-2/m$ derived in prior work for $\alpha=1$~\cite{gross2017vote,goel2017metric}.

\begin{theorem}\label{thm:smartdict-opt-plurality-scores}
	For every $m \geq 2$ and $\alpha \in [0,1]$, no randomized social choice rule with access to only plurality scores has distortion better than $2+\alpha-2/m$ for $\alpha$-decisive metric spaces.
\end{theorem}
\begin{proof}
Let $f$ be any randomized social choice rule. Fix any $m$. Consider an instance with $n=m$ voters in which every candidate has plurality score $1$, i.e., $\plu(c) = 1$ for each $c \in C$. By the pigeonhole principle, there exists a candidate $c^*$ that $f$ picks with probability at most $1/m$ on this instance. 

We now construct an underlying metric space consistent in which each candidate has plurality score $1$, but choosing $c^*$ with probability at most $1/m$ leads to bad distortion. 

For the voter $i$ whose top choice is $c^*$, we let $d(i,c^*)=0$ and $d(i,c)=1$ for every $c\neq c^*$. For every other voter $i$, we let $d(i,\top(i)) = \alpha$, $d(i,c^*) = 1+\alpha$, and $d(i,c) = 2+\alpha$ for all $c \neq \top(i),c^*$. It is easy to check that this forms a metric (let all other distances be the induced ``shortest path'' distances). Note that the social cost of candidate $c^*$ is
$\SC(c^*) = \sum_{i\in V} d(i,c^*) = m-1$, whereas for any other candidate $b \neq c^*$, 
\begin{align*}
\SC(b) &= \sum_{i\in V} d(i,b) = \sum_{i\in V: \top(i)=b} d(i,b) + \sum_{i\in V: \top(i)=c^*} d(i,b) + \sum_{i\in V: \top(i)\neq b,c^*} d(i,b)\\
&= \alpha +  (1+\alpha) + (m-2) \cdot (2+\alpha).
\end{align*}

It is easy to verify that for all $m\geq 2$ we have $\SC(c^*) \le \SC(b)$ for all $b\neq c^*$. Since $f$ chooses $c^*$ with probability at most $1/m$ and candidates from $C\setminus \{c^*\}$ with the remaining probability of at least $(m-1)/m$, its distortion is at least
\begin{align*}
\frac{\frac{1}{m} \cdot \SC(c^*)+ \frac{m-1}{m} \cdot \SC(b)}{\SC(c^*)} &= \frac{(m-1)+(m-1)(\alpha+(1+\alpha)+(m-2)(2+\alpha))}{m \cdot (m-1)} \\
&= \frac{2m+m\alpha-2}{m}\\
&= 2+\alpha-\frac{2}{m},
\end{align*}
as desired.
\end{proof}

In terms of general lower bounds, when the randomized rule is given access to the full preference profile, \citet{anshelevich2017randomized} showed that no randomized rule can achieve distortion better than $1+\alpha$ for $m \ge 2$. Our next result provides an improved bound, which matches the bound of $1+\alpha$ for $m=2$ (which is optimal for $m=2$), but increases to $(3+\alpha)/2$ as $m \to \infty$.

\begin{theorem}\label{thm:randomized-lower-bound}
	For every $m \ge 2$ and $\alpha \in [0,1]$, no randomized social choice rule has distortion better than $\frac{3+\alpha}{2} - \frac{1-\alpha}{\floor{m}_{even}}$ for $\alpha$-decisive metric spaces, where $\floor{m}_{even}$ is the largest even integer that smaller than or equal to $m$. 
\end{theorem}
\begin{proof}
We use the same construction that we used in the proof of \Cref{thm:det-lower-bound}. For even $m$, recall that the construction had two equal-sized groups of candidates $A$ and $B$. There were two possible consistent $\alpha$-decisive metric spaces, one in which each candidate in $A$ is optimal and each candidate in $B$ is $2+\alpha-2(1-\alpha)/m$ times worse, and the other in which the roles of $A$ and $B$ are reversed. 

Due to the symmetry, assume that a given randomized rule places a total probability of at least $0.5$ on candidates from $A$. Then, in the latter metric space, the rule chooses an optimal candidate with probability at most $0.5$ and a candidate with approximation ratio $2+\alpha-2(1-\alpha)/m$ with the remaining probability of at least $0.5$. Therefore, its distortion is at least 
\[\frac{1}{2} \left(2+\alpha-\frac{2(1-\alpha)}{m}\right)+\frac{1}{2} = \frac{3+\alpha}{2} - \frac{1-\alpha}{m}.\]

For odd $m$, we use the same trick as in the proof of \Cref{thm:det-lower-bound}: we use the construction for $m-1$, and adding a candidate sufficiently far away, which leads to the claimed bound.
\end{proof}

One of the most compelling open problems is whether there exists a randomized rule that can break the barrier of $2+\alpha$ even as $m \to \infty$. When $m\to \infty$, the distortion of $2+\alpha$ can be achieved by both \randdict and \smartdict in a randomized fashion, as well as deterministically by \plumatching. \Cref{thm:randomized-lower-bound} implies that a distortion bound better than $\frac{3+\alpha}{2}$ would not be possible in this case, leaving some room for improvement. Of particular interest is the case of $\alpha=1$ for which the best possible distortion is in $[2,3]$. 

In order to achieve a bound better than $2+\alpha$, we need a rule that is randomized and, as Theorem~\ref{thm:smartdict-opt-plurality-scores} shows, accesses more information regarding the preference profile than just the plurality scores of the candidates. A tempting thought is that an appropriate (randomized) combination of a known randomized rule with our deterministic \plumatching rule, which depends on a lot more information beyond the plurality scores, could be successful. The following lower bound shows that this approach cannot yield improved bounds when \plumatching is combined with a randomized rule that only places positive probability on candidates with positive plurality score (as both \randdict and \smartdict do). This suggests that any successful solution would need to, at least for some instances, assign positive probability to candidates that have zero plurality score and do not satisfy the conditions of the \plumatching rule.

\begin{theorem}\label{thm:mix_lb}
For any $\alpha \in [0,1]$, every randomized social choice rule that randomizes over only candidates that could be returned by \plumatching and candidates with non-zero plurality score has distortion at least $2+\alpha-2/(m-2)$ for $\alpha$-decisive metric spaces.
\end{theorem}
\begin{proof}
We consider an instance with a set of $m=\ell +2$ candidates, $\{a_1, a_2, \dots, a_{\ell+1}, a^*\}$, and a set of voters that are partitioned into $\ell+1$ subsets, $V_1, V_2, \dots, V_{\ell+1}$. The set $V_{\ell+1}$ contains a single voter, and each of the other $\ell$ sets contain exactly $k$ voters, for some arbitrarily large $k\in \mathbb{N}$, so the total number of voters is $n=k\ell+1$. The distances of the voters from the candidates are implied by shortest path distances in the graph of Figure~\ref{fig:LP-for-mix}, so for the voters of each group $V_i$ for $i \leq \ell$, their ranking is $a_i$ first, $a^*$ second, and $a_{\ell+1}$ last, with some arbitrary ordering of the remaining candidates in-between $a^*$ and $a_{\ell+1}$. On the other hand, the ranking of the single voter in $V_{\ell+1}$ is $a_{\ell+1}$ first, $a^*$ last, and some arbitrary ordering of the remaining candidates in-between.

Note that candidate $a^*$ has the optimal social cost, with $\SC(a^*)=k\ell+3+\alpha$, but the plurality score of this candidate is zero. So, this candidate will never be chosen by a social choice rule that randomizes over only candidates with non-zero plurality score. Also, note that the voter in $V_{\ell+1}$ ranks candidate $a^*$ last, and no voter ranks $a^*$ first. In other words it's \emph{veto score} is less than its plurality score. We show in \Cref{lem:plu-veto} in \Cref{sec:condorcet} that this implies $a^*$ cannot be chosen by \plumatching.

Given the symmetry of the instance, it is easy to verify that the social cost among the remaining candidates is at least as much as the social cost of $a_1$. The social cost of this candidate is $\SC(a_1)=k \alpha + (\ell-1)k(2+\alpha) + 2=\ell k \alpha +2(\ell-1)k+2$. We therefore conclude that the distortion for this instance will be at least
\[\frac{\SC(a_1)}{\SC(a^*)}= \frac{\ell k \alpha +2(\ell-1)k+2}{k\ell+3+\alpha},\]
 which for $k\to \infty$ becomes
\[\frac{\ell \alpha +2(\ell-1)}{\ell} =2+\alpha-\frac{2}{\ell}=2+\alpha-\frac{2}{m-2}.\qedhere\]
\end{proof}

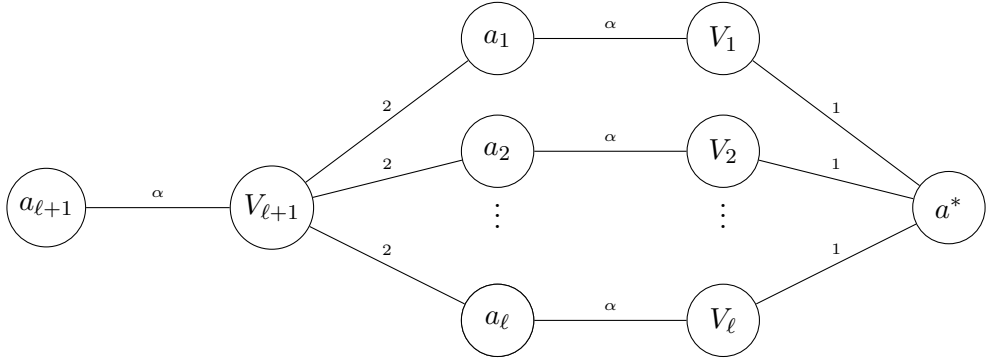
\begin{figure}
        \begin{center}
            \begin{tikzpicture}[scale=1.5]
                \tikzstyle{main_node} = [circle,fill=white,draw,minimum size=2.5em,inner sep=3pt]

                \node[main_node] (a1) at (0,0) {$a_1$};
                \node[main_node] (a2) at (0, -1) {$a_2$};
                \node at (0, -1.5) {\vdots};
                \node[main_node] (al) at (0, -2.5) {$a_\ell$};
                \node[main_node] (al+1) at (0, -2.5) {$a_\ell$};

                \node[main_node] (va1) at (2,0) {$V_1$};
                \node[main_node] (va2) at (2, -1) {$V_2$};
                \node at (2, -1.5) {\vdots};
                \node[main_node] (val) at (2, -2.5) {$V_\ell$};

                \node[main_node] (p) at (4, -1.5) {$a^*$};
                \node[main_node] (val+1) at (-2, -1.5) {$V_{\ell+1}$};
                \node[main_node] (al+1) at (-4, -1.5) {$a_{\ell+1}$};

                \draw (p) -- (va1) node [midway, above] {\tiny $1$};
                \draw (p) -- (va2) node [midway, above] {\tiny $1$};
                \draw (p) -- (val) node [midway, above] {\tiny $1$};

                \draw (val+1) -- (a1) node [midway, above] {\tiny $2$};
                \draw (val+1) -- (a2) node [midway, above] {\tiny $2$};
                \draw (val+1) -- (al) node [midway, above] {\tiny $2$};
                \draw (val+1) -- (al+1) node [midway, above] {\tiny $\alpha$};

                \draw (a1) -- (va1) node [midway, above] {\tiny $\alpha$};
                \draw (a2) -- (va2) node [midway, above] {\tiny $\alpha$};
                \draw (al) -- (val) node [midway, above] {\tiny $\alpha$};

            \end{tikzpicture}
        \end{center}
        \caption{The metric space construction used in the proof of \Cref{thm:mix_lb}.}
        \label{fig:LP-for-mix}
    \end{figure}
\section{Other Related Work}\label{sec:related}
The parametrization using $\alpha$-decisiveness provides an interesting restriction on the underlying metric space, and the special case of $\alpha=0$ captures settings where the sets of voters and candidates coincide, i.e., $V=C$. This is known as \emph{peer selection}, which involves a set of individuals voting over themselves. Typically, peer selection is studied under binary votes with a focus on strategic manipulation~\cite{AFPT11}, whereas our focus is on determining the desired selection rule (disregarding any manipulations) under ranked votes. The case where $V=C$ is also widely studied in theoretical computer science under the name of the \emph{metric 1-median problem}: given a finite metric space $(\calM,d)$, find a point in $\calM$ with the least total distance to all points in $\calM$ (think $V = C = \calM$). Given all distances, the problem can be solved trivially in $O(|\calM|^2)$ time, and a $(1+\epsilon)$-approximation can be computed in $O(|\calM|/\epsilon^2)$ time through Monte-Carlo search~\cite{indyk1999sublinear,indyk2000high}. A different line of inquiry limits the number of pairwise distances the algorithm can query~\cite{chang2012some,chang2017lower,chang2017metric}; like in our metric distortion problem, the goal is to find the best approximation given limited information. 

Prior work in voting has also considered imposing other well-motivated restrictions on the preferences of the voters, e.g., by considering Euclidean metrics~\cite{sui2013multi,anshelevich2017randomized}, as well as metrics inducing single-peaked~\cite{Moul80} or single-crossing~\cite{gans1996majority,bredereck2013characterization} preferences. An alternative way of introducing more structure to the voter preferences is to assume the candidates are \emph{representative} of the voters, i.e., drawn i.i.d.\ from the pool of voters~\cite{cheng2017representative,cheng2017people}. Note that in this case, each candidate is co-located with some voter, which is reminiscent of the special case $\alpha=0$ discussed above. \citet{pierczynski2019approval} extend the idea of distortion to approval-based preferences, while \citet{BLSS19} and \citet{filos2019distortion} extend it to complex voting paradigms such as district elections and primaries. On the other end, prior work has also studied distortion in less constrained settings, e.g., ones that are not even bound by metric constraints~\cite{PR06,BCHL+15,CNPS17}.

Distortion minimization, at its core, is the simple idea of designing approximation algorithms that have access only to ordinal information about underlying cardinal values. One can therefore leverage the distortion framework to study a diverse set of problems or objectives. As an alternative objective, \citet{fain2019random} focus on minimizing \emph{squared distortion} of randomized rules, which is a notion that helps bound both the expectation and the variance of the approximation ratio of a rule at the same time. \citet{goel2017metric,goel2018relating} optimize the fairness ratio, which is a notion that is closely related to and always at least the distortion (in Appendix~\ref{sec:fairness}, we show that a simple modification of our distortion analysis can also be used to achieve a bound on the fairness ratio of our deterministic rule). The distortion framework has also already been used to study a variety of different problems: \citet{anshelevich2016blind,anshelevich2016truthful}, \citet{abramowitz2018graph}, and \citet{anshelevich2019tradeoff,anshelevich2018facility} design novel algorithms for problems such as matching, clustering, and facility location using only ordinal information. For the case of matching, recent work by \citet{ACGH20} also evaluated ordinal mechanisms compared to cardinal ones in strategic settings aiming to design truthful mechanisms.

Recent work has also considered moving beyond the cardinal versus ordinal comparison, aiming to get a more refined understanding of the value of information in different settings. \citet{chen2019favorite,fain2019random,gross2017vote} are interested in rules that elicit \emph{even less} information from the voters than their full ranked preferences. On the contrary, \citet{abramowitz2019passion} consider rules that have access to \emph{more} information than just the ranked preferences, specifically the ratio of distances of each voter to every pair of candidates. Extending this idea, \citet{kempe2020communication}, \citet{MPSW19}, \citet{MSW20}, and \citet{ABFV20} study the explicit tradeoff between communication requirements and distortion of voting rules. 
More broadly, the \emph{quantitative} evaluation of social choice rules using distortion, arguably a contribution of computer science to voting theory~\cite{PR06,BCHL+15}, is a recent approach, and stands in stark contrast to centuries-long study of social choice rules through an \emph{axiomatic} approach~\cite{BCEL+16}.

\section{Discussion}\label{sec:disc}

In this paper, we settled the optimal metric distortion conjecture by proving that the optimal distortion of deterministic social choice rules is $3$, and introducing a novel polynomial time computable rule \plumatching that achieves this. 

While the bound of $3$ is tight for general metric spaces, for the restrictive class of $\alpha$-decisive metric spaces we showed that \plumatching achieves distortion $2+\alpha$, and there is a lower bound of $2+\alpha-2(1-\alpha)/m$ for all deterministic rules. While this establishes \plumatching as optimal for all $\alpha$ when $m \to \infty$, and for all $m$ when $\alpha = 1$, it leaves open the question of bridging the bounds when $m$ is finite and $\alpha < 1$.  

On the randomized frontier, we proposed \smartdict, which achieves distortion $2+\alpha-2/m$, improving over the best known bound of $2+\alpha$ achieved by \randdict. This also established a separation between the power of randomized and deterministic rules when $\alpha > 0$. However, determining whether randomized rules can outperform deterministic rules for $\alpha = 0$, and analyzing the optimal distortion of randomized rules remain challenging open questions. 

Finally, the core strength behind our deterministic \plumatching rule is the novel Ranking-Matching Lemma we prove about matching rankings to candidates. This is an elementary lemma that could be of interest outside of the metric distortion framework as it induces a novel family of voting rules, some of which may have other interesting properties. For example, we defined the \unimatching rule in \Cref{sec:rm-lemma}, which returns a candidate satisfying the following guarantee: \emph{In any subset of $x\%$ of the votes, for any $x$, the selected (single) candidate weakly defeats at least $x\%$ of the candidates.} 

This guarantee is eerily reminiscent of another major open question in voting theory. Consider selecting a subset, or \emph{committee}, $W$ of $k$ candidates instead of a single candidate. This committee $W$ is said to be in the \emph{core} if no group $S$ of at least $n/k$ voters can find a candidate $c^*$ such that each voter $i \in S$ prefers $c^*$ to every candidate in the selected committee $W$. It is easy to check that this can be alternatively phrased as: \emph{In any subset of $n/k$ of the votes, the selected $k$ candidates weakly defeat all (i.e. at least $m$) candidates.} 

Thus, instead of finding one candidate that weakly defeats at least $m/k$ candidates in any $n/k$ votes, we want to find $k$ candidates that (collectively) weakly defeat all $m$ candidates in any $n/k$ votes. While it is known that such a committee does not always exist~\cite{jiang2020approximately}, the question remains open if we relax the size of the group from $n/k$ voters to $2n/k$ voters. Similarly, the question of whether a committee in the core always exists is also open (without any relaxations required) when voters have binary preferences over candidates~\cite{fain2018fair}. It is natural to think of recursively applying rules such as \unimatching or \plumatching to select committees. While the distortion of such a recursive application is easy to bound~\cite{goel2018relating}, whether it provides any core-like guarantees is an important question for the future.  
\appendix
\section*{Appendix}

Below, we provide the missing proofs as well as some additional observations. 

\section{Proof of \Cref{lem:general-halls-condition}}\label{sec:proof-halls}
\begin{proof}
	Let $\elec = (V,C,\vsigma)$ be an election, $p \in \Delta(V)$ and $q \in \Delta(C)$ be weight vectors, and $a \in C$ be a candidate. First, we want to show that $G^{\elec}_{p,q}(a)$ admits a fractional perfect matching if and only if $q(\defvotese{a}{S}{\elec}) \ge p(S)$ for all $S \subseteq V$. 
	
	This can be proved, similarly to the regular Hall's condition, through the max-flow min-cut theorem. Create a flow network $N$ shown in \Cref{fig:flow}. Note that we have converted all edges of $G^{\elec}_{p,q}(a)$ to directed edges from voters to candidates with infinite capacity, added a source node with an outgoing edge of capacity $p(i)$ to each voter $i$, and added a target node with an incoming edge of capacity $q(c)$ from each candidate $c$. 
	
	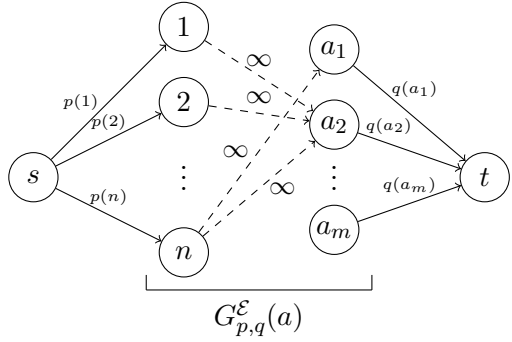
\begin{figure}
		\centering
		\begin{tikzpicture}
		\centering
		\tikzstyle{main_node} = [circle,fill=white,draw,minimum size=1.7em,inner sep=0pt]
		
		\node[main_node] (s) at (-2, 0) {$s$};
		\node[main_node] (1) at (0, 2) {$1$};
		\node[main_node] (2) at (0, 1) {$2$};
		\node[main_node] (n) at (0, -1) {$n$};
		\node[main_node] (a1) at (2, 1.7) {$a_1$};
		\node[main_node] (a2) at (2, .7) {$a_2$};
		\node[main_node] (am) at (2, -.7) {$a_m$};
		\node[main_node] (t) at (4, 0) {$t$};
		
		\node at (0, .1) {$\vdots$};
		\node at (2, .1) {$\vdots$};
		
		\draw[->] (s)--(1) node [midway,  left] {\tiny$p(1)$};
		\draw[->] (s)--(2) node [midway, above] {\tiny$p(2)$};
		\draw[->] (s)--(n) node [midway, above] {\tiny$p(n)$};
		\draw[->] (a1)--(t) node [near start, right] {\tiny$q(a_1)$};
		\draw[->] (a2)--(t) node [pos=.3, above] {\tiny$q(a_2)$};
		\draw[->] (am)--(t) node [midway, above] {\tiny$q(a_m)$};
		
		\draw[->, dashed] (1)--(a2) node [midway, above] {\small $\infty$};
		\draw[->, dashed] (2)--(a2) node [midway, above] {\small $\infty$};
		\draw[->, dashed] (n)--(a2) node [midway, right] {\small $\infty$};
		\draw[->, dashed] (n)--(a1) node [midway, left] {\small $\infty$};
		
		\draw (-.5,-1.3)--(-.5,-1.5);
		\draw (-.5, -1.5)--(2.5,-1.5) node [midway, below] {$G^\elec_{p,q}(a)$};
		\draw (2.5,-1.5)--(2.5,-1.3);
		
		\end{tikzpicture}
		\caption{Flow network used to prove a generalization of the Hall's condition in \Cref{lem:general-halls-condition}. Edges from source node $s$ to each voter $i$ has capacity $p(i)$, edges from each candidate $c$ to target $t$ has capacity $q(c)$, and edges from voters to candidates are as in $G^{\elec}_{p,q}(a)$ with infinite capacity.}
		\label{fig:flow}
	\end{figure}

	We first argue that this network admits a flow of value $1$ if and only if $G^{\elec}_{p,q}(a)$ admits a fractional perfect matching. If $N$ has a flow of value $1$, then $s$ has an outgoing flow of $1$ and $t$ has an incoming flow of $1$. But since $\sum_{i \in V} p(i) = \sum_{c \in C} q(c) = 1$, this means the node for each voter $i$ must have a flow of $p(i)$ passing through it, and the node of each candidate $c$ must have a flow of $q(c)$ passing through it. A fractional perfect matching in $G^{\elec}_{p,q}(a)$ can now be obtained by using the flow on each $(i,c)$ edge as the weight of the edge. Conversely, if $G^{\elec}_{p,q}(a)$ admits a fractional perfect matching $w$, then it is straightforward to see that having each $(s,i)$ edge carry a flow of $p(i)$, each $(i,c)$ edge carry a flow of $w(i,c)$, and each $(c,t)$ edge carry a flow of $q(c)$, for every $i \in V$ and $c \in C$, results in a valid flow of value $1$. 
	
	Next, because $p$ and $q$ are normalized weight vectors, the network has a flow of value $1$ if and only if its maximum flow is $1$. By the max-flow min-cut theorem, this is equivalent to its minimum cut capacity being $1$. Since $(\set{s},V \cup C \cup \set{t})$ is a trivial cut with capacity $1$, min cut being equal to $1$ is equivalent to every cut having capacity at least $1$. 
	
	We now show that every cut in this network has capacity at least $1$ if and only if $q(\defvotese{a}{S}{\elec}) \ge p(S)$ for each $S \subseteq V$. For the forward direction, suppose each cut has capacity at least $1$. Consider the cut $(\set{s} \cup S \cup \defvotese{a}{S}{\elec}, (V \setminus S) \cup (C \setminus \defvotese{a}{S}{\elec}) \cup \set{t})$. Because $\defvotese{a}{S}{\elec}$ is precisely the set of neighbors of $S$, there are no edges of $\infty$ capacity in this cut. The only edges part of the cut are those going from $s$ to nodes in $V \setminus S$ and those going from nodes in $\defvotese{a}{S}{\elec}$ to $t$. Thus, the capacity of this cut is $p(V \setminus S) + q(\defvotese{a}{S}{\elec})$, which should be at least $1$. Using $p(S) = 1-p(V\setminus S)$, we get the desired condition. 
	
	For the reverse direction, suppose $q(\defvotese{a}{S}{\elec}) \ge p(S)$ for each $S \subseteq V$. Consider any cut $(\set{s} \cup A, B \cup \set{t})$ of the network. If it contains any edge of $\infty$ capacity, then its capacity is clearly at least $1$. If it does not, then letting $S = A \cap V$, all neighbors of nodes in $S$ must also be in $A$. Hence, $\defvotese{a}{S}{\elec} \subseteq A$. Thus, the cut has capacity at least $p(V\setminus S) + q(\defvotese{a}{S}{\elec}) = 1-p(S)+q(\defvotese{a}{S}{\elec}) \ge 1$, as desired.
	
	We have thus established that $G^{\elec}_{p,q}(a)$ admits a fractional perfect matching if and only if the network constructed has max flow $1$, if and only if $q(\defvotese{a}{S}{\elec}) \ge p(S)$ for each $S \subseteq V$.
	
	Finally, since maximum flow value can be computed in strongly polynomial time, it follows that the existence of a fractional perfect matching in $G^{\elec}_{p,q}(a)$ can also be checked in strongly polynomial time.
\end{proof}

We remark that while \Cref{lem:general-halls-condition} is stated in terms of $(p,q)$-domination graphs, it in fact applies to all vertex-weighted bipartite graphs and provides a Hall's-like condition for the existence of a fractional perfect matching in such graphs.

\section{\plumatching and Condorcet Consistency}\label{sec:condorcet}

In \Cref{subsec:connection-to-the-mus}, we showed that \plumatching is not Condorcet consistent, but this is necessary in order for it to achieve a distortion bound of $2+\alpha$, as no Condorcet consistent rule can achieve this bound (\Cref{prop:condorcet-bad}). The preference profile used in the proof of \Cref{prop:condorcet-bad} is an example where \plumatching does not select a Condorcet winner, even when one exists.

Here, we present a simpler and more insightful example in which the same phenomenon occurs. First, we observe an important necessary condition for a candidate to be potentially returned by \plumatching (i.e.\ for its integral domination graph to have a perfect matching). This condition is easy to check, and we use it elsewhere in the paper to rule out various candidates from being the output of \plumatching.

Recall that $\plu(a)$ denotes the plurality score of $a$, i.e., the number of voters who rank $a$ first. Define $\veto(a)$ to be the veto score of $a$, i.e., the number of voters who rank $a$ last. 

\begin{lemma}\label{lem:plu-veto}
	Given an election $\elec = (V,C,\elec)$ and candidate $a \in C$, $a$ can be returned by \plumatching (i.e.its integral domination graph $G^{\elec}(a)$ has a perfect matching) only if $\plu(a) \ge \veto(a)$.
\end{lemma}
\begin{proof}
	Note that in the integral domination graph $G(a)$, every $i \in V$ on the left who ranks $a$ last --- there are $\veto(a)$ many such voters on the left --- must be matched to some $j \in V$ on the right who ranks $a$ first --- there are $\plu(a)$ many such voters on the right. Hence, for a perfect matching to exist, we must have $\veto(a) \le \plu(a)$, as desired. 
\end{proof}

It is easy to see that a Condorcet winner may not necessarily satisfy this. In fact, it may have zero plurality score and non-zero veto score, as the example below shows.

\begin{example}
	Consider an election $\elec$ with a set of $7$ voters $V = \set{1,\ldots,7}$, a set of $4$ candidates $C = \set{a,b,c,d}$, and the following preference profile $\vsigma$.
	\[\begin{array}{ll}
	\sigma_1 = \sigma_2 &: b \succ a \succ c \succ d\\
	\sigma_3 = \sigma_4 &: c \succ a \succ b \succ d\\
	\sigma_5 = \sigma_6 &: d \succ a \succ b \succ c\\
	\sigma_7 &: b \succ c \succ d \succ a
	\end{array}\]

	It is easy to check that $a$ is the only Condorcet winner (in fact, for every other candidate, there is a \emph{strict} majority of voters preferring $a$ over that candidate). However, $\plu(a) = 0 < \veto(a) = 1$. Thus, by \Cref{lem:plu-veto}, $a$ cannot be picked by \plumatching. 
\end{example}

\section{\plumatching and Tie-Breaking}\label{sec:ties}

\Cref{thm:det-upper-bound} established that the distortion of \plumatching is $2+\alpha$ for $m \ge 3$. The tightness of the bound was shown through an example in which some candidate $a$ whose integral domination graph had a perfect matching (and thus \emph{could} be picked by \plumatching) had $2+\alpha$ approximation to social cost. However, in that example, even the integral domination graph of an optimal candidate had a perfect matching, and therefore could also be picked by \plumatching, leading one to hope that perhaps wisely picking among candidates whose integral domination graphs have a perfect matching may lead to lower distortion. Below, we show that the distortion remains no less than $2+\alpha$ regardless of the tie-breaking scheme used. 

\begin{proposition}\label{prop:tie-breaking}
	For every $m \geq 3$ and $\alpha \in [0, 1]$, every refinement of \plumatching (i.e.\ every deterministic social choice rule that always returns a candidate whose integral domination graph admits a perfect matching) has distortion $2 + \alpha$.
\end{proposition}
\begin{proof}
	The upper bound follows from \Cref{thm:det-upper-bound}. For the lower bound, we construct a family of elections with three candidates and an increasing number of voters such that (i) in each election, there is only a single candidate that could be chosen by \plumatching, and (ii) the approximation ratio of this candidate converges to $2+\alpha$ as the number of voters grows. 
	
	For $k \geq 1$, we construct an election with $2k + 2$ voters ($V = \set{1, \ldots, 2k + 2})$ and three candidates ($C = \set{a, b, c}$). The preference profile is as follows: for each $1 \leq i \leq k$, $\sigma_i = a \succ b \succ c$, for each $k + 1 \leq i \leq 2k$, $\sigma_i = c \succ b \succ a$, $\sigma_{2k+1} = b \succ a \succ c$, and $\sigma_{2k+2} = b \succ c \succ a$. 
	
	First, note that both $a$ and $c$ have veto score $k + 1$ but plurality score only $k$. Therefore, by \Cref{lem:plu-veto}, neither can be chosen by \plumatching. Therefore, \plumatching must pick candidate $b$. However, consider the $\alpha$-decisive metric consistent with the preference profile given by the undirected graph in \Cref{fig:tie-breaking-metric}; as usual, all pairwise distances are induced by shortest paths. 
	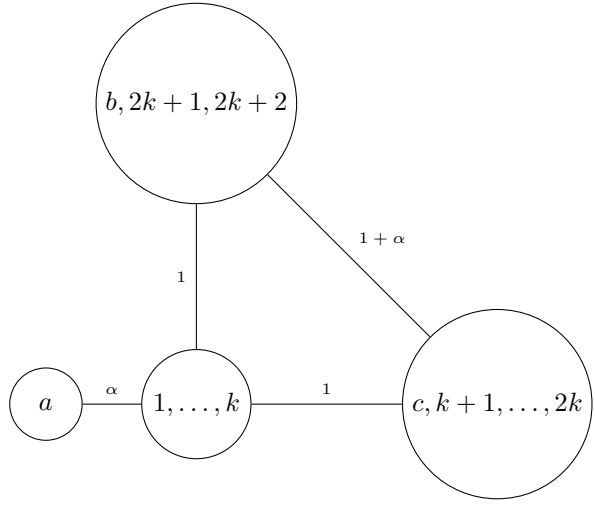
\begin{figure}
		\centering
		\begin{tikzpicture}[scale=2]
		\tikzstyle{main_node} = [circle,fill=white,draw,minimum size=2.5em,inner sep=3pt]
		
		\node[main_node] (a) at (0,0) {\small $a$};
		\node[main_node] (1k) at (1, 0) {\small $1,\ldots,k$};
		\node[main_node] (b) at (1, 2) {\small $b, 2k+1, 2k + 2$};
		\node[main_node] (c) at (3, 0) {\small $c, k + 1, \ldots, 2k$};
		
		\draw (a) -- (1k) node [midway, above] {\tiny $\alpha$};
		\draw (1k) -- (b) node [midway, left] {\tiny $1$};
		\draw (1k) -- (c) node [midway, above] {\tiny $1$};
		\draw (b) -- (c) node [midway, above right] {\tiny $1 + \alpha$};
		\end{tikzpicture}
		\caption{The metric space construction  used in the proof of \Cref{prop:tie-breaking}}
		\label{fig:tie-breaking-metric}
	\end{figure}

	Note that $\SC(b) = k \cdot (2 + \alpha)$, while $\SC(c) = k + 2\cdot (2 + \alpha)$. Hence, choosing $b$ yields approximation ratio 
	\begin{align*}
	\frac{k(2 + \alpha)}{k + 2(2 + \alpha)}
	&= \frac{k(2 + \alpha) + 2(2+\alpha)(2 + \alpha) - 2(2+\alpha)(2 + \alpha)}{k + 2(2 + \alpha)}\\
	&= \frac{(k + 2(2 + \alpha))(2 + \alpha) - 2(2+\alpha)^2}{k + 2(2 + \alpha)}\\
	&= 2 + \alpha - \frac{2(2+\alpha)^2}{k + 2(2 + \alpha)}.
	\end{align*}
	
	As $k$ grows, this approaches $2 + \alpha$. Hence, the distortion of $f$ is at least $2+\alpha$ when $m = 3$. For $m > 3$, we can add remaining candidates sufficiently far away as before, which does not affect the construction. 
\end{proof}

\section{\plumatching and the Fairness Ratio}\label{sec:fairness}

\citet{goel2017metric} observe that given an underlying metric and a chosen candidate, costs may vary widely among the voters. In order to quantify how ``fair'' a candidate is, they look at a variation of the social cost, defined, for any given $k \in \set{1,\ldots,n}$, as the sum of the $k$ largest costs to voters. In particular, given a metric $d$, candidate $c \in C$, and $k \in \set{1,\ldots,n}$, they define
\[
\phi_k(c, d) = \max_{S \subseteq V: |S|=k} \sum_{v \in V} d(v, c).
\]

Then, they define the \emph{fairness ratio} of a social choice rule $f$, denoted $\fair(f)$, as its worst-case approximation ratio to the expected total of the $k$ largest voter costs, for any $k$. That is, 
\[
\fair(f) = \sup_{\vsigma}\ \sup_{d\: :\: d\: \triangleright\: \vsigma}\ \max_{1 \le k \le n}\ \frac{\E[\phi_k(f(\vsigma), d)]}{\min_{c \in C} \phi_k(c, d)}.
\]

Since fixing $k=n$ yields the distortion definition, it follows that $\dist(f) \le \fair(f)$ for any social choice rule $f$. Interestingly, \citet{goel2018relating} show that $\fair(f) \le \dist(f)+2$ for any social choice rule $f$. Hence, the fairness ratio is within an additive factor of $2$ from the distortion. 

Given that \plumatching has distortion at most $2+\alpha$ for $m \ge 3$ (\Cref{thm:det-upper-bound}), it immediately follows that its fairness ratio is at most $4+\alpha$. However, a simple modification of the proof of \Cref{thm:det-upper-bound} shows that its fairness ratio is in fact $2+\alpha$ as well, which is a stronger result than its distortion being $2+\alpha$. 

\begin{proposition}\label{prop:fairness-ratio}
	For every $m \ge 2$ and $\alpha \in [0,1]$, \plumatching has fairness ratio $2+\alpha$ for $\alpha$-decisive metric spaces.
\end{proposition}
\begin{proof}
	Since $\fair(\plumatching) \ge \dist(\plumatching)$, a lower bound of $2+\alpha$ follows directly from \Cref{thm:det-upper-bound}. For the upper bound, the proof is nearly identical to the proof of $2 + \alpha$ upper bound on the distortion in \Cref{thm:det-upper-bound}. The trick is to observe that when we start from $\phi_k(a,d)$ for the candidate $a$ picked by \plumatching, it consists of the $k$ largest costs of $a$ to some $k$ voters. As we manipulate this quantity to express it in terms of the cost of some other candidate $b$, we get an expression involving the cost of $b$ to \emph{some} $k$ voters, which is at most $\phi_k(b,d)$. Crucially, we use the fact that there are precisely $k$ voters matched to any $k$ voters in a perfect matching in the integral domination graph. We reproduce the full proof below for completeness. 
	
	Let $(\calM,d)$ be an $\alpha$-decisive metric space, and
	let $\vsigma$ be the induced preference profile. Let $a$ be the candidate returned by \plumatching, and let $b$ be
	any other candidate. Fix some $k \in \set{1,\ldots,n}$. We want to show that $\phi_k(a, d) \le (2+\alpha) \cdot \phi_k(b, d)$.
	Let $S_a, S_b \subseteq V$ be sets of voters of size $k$ that have the $k$ largest distances to candidates $a$ and $b$, respectively. Due to \Cref{cor:integral-domination}, the integral domination graph $G^{\elec}(a)$ admits a perfect matching $M : V \to V$ which satisfies $a \succeq_i \top(M(i))$ for each $i \in V$. Thus,
    \begin{align*}
        \phi_k(a, d) &= \sum_{i \in S_a} d(i,a)\\
        &\le \sum_{i \in S_a} d(i,\top(M(i))) &(\because a \succeq_i \top(M(i)), \forall i \in V)\\
        &\le \sum_{i \in S_a} \big( d(i,b) + d(b,\top(M(i))) \big) &(\because \text{triangle inequality})\\
        &= \sum_{i \in S_a} d(i, b) + \sum_{i \in S_a} d(b,\top(M(i)))\\
        &= \sum_{i \in S_a} d(i, b) + \sum_{i \in M(S_a)} d(b,\top(i)) &(M(S_a) = \set{M(i) : i \in S_a})\\
        &= \sum_{i \in S_a} d(i, b) + \sum_{i \in M(S_a) : \top(i) \neq b} d(b,\top(i)) &(\because d(b,b)=0)\\
        &\le \sum_{i \in S_a} d(i, b) + \sum_{i \in M(S_a) : \top(i) \neq b} \left(d(b,i)+d(i,\top(i))\right) &(\because \text{triangle inequality})\\
        &\le \sum_{i \in S_a} d(i, b) + \sum_{i \in M(S_a) : \top(i) \neq b} \left(d(b,i)+\alpha \cdot d(i,b)\right) &(\because \text{$\alpha$-decisiveness})\\
        &\le \sum_{i \in S_a} d(i, b) + \sum_{i \in M(S_a)} \left(d(b,i)+\alpha \cdot d(i,b)\right)\\
		&\le (2 + \alpha) \phi_k(b, d),
    \end{align*}
    where the last inequality follows because $\sum_{i \in S} d(i, b) \leq \phi_k(b, d)$ for any $S$ with $|S| = k$, and $M$ being a perfect matching, we have $|S_a| = |M(S_a)| = k$. 
    
    Because the metric space, the election, $k$, and the choice of $b$ were arbitrary, this establishes that \plumatching has fairness ratio $2+\alpha$.
\end{proof}

\bibliographystyle{plainnat}
\bibliography{abb,ultimate}

\end{document}